\documentclass[sigconf]{acmart}
\usepackage[utf8]{inputenc}
\usepackage{textgreek}
\usepackage{xcolor}
\usepackage{amsmath}
\usepackage[ruled, vlined, linesnumbered]{algorithm2e}
\usepackage{bm}
\usepackage{enumitem}
\usepackage{lipsum} 
\usepackage[most]{tcolorbox}
\usepackage{tikz}
\usepackage{pgfplots}
\usetikzlibrary{arrows.meta, positioning, backgrounds, fit}
\usepackage[aboveskip=-4pt]{subcaption}
\usepgfplotslibrary{fillbetween}
\usepackage{amsthm}
\usepackage{mdframed}

\usepackage{tabularx}
\usepackage{array}
\newcolumntype{R}{>{\raggedleft\arraybackslash}X} 

\theoremstyle{definition}
\newtheorem{definition}{Definition}

\newcommand{\eat}[1]{} 

\newcommand{\stitle}[1]{\noindent{\bf #1.\/}}

\newtheorem{corollary}{Corollary}
\newtheorem{problem}{Problem}

\newtheorem{assumption}{Assumption}

\definecolor{hlbg}{RGB}{200,240,240}

\newcommand{\ourproblem}{constrained LLM selection for compound AI systems\xspace}
\newcommand{\ours}{\textsf{SCOPE}\xspace}
\newcommand{\gours}{\textsf{SCOPE-Coarse}\xspace}
\newcommand{\ggours}{\textsf{SCOPE-Rand}\xspace}
\newcommand{\warmstart}{\textsf{Calibrate}\xspace}

\newcommand{\abacus}{\textsf{Abacus}\xspace}

\newcommand{\BO}{\textsf{cEI}\xspace}
\newcommand{\config}{\textsf{CONFIG}\xspace}
\newcommand{\llambo}{\textsf{LLAMBO}\xspace}
\newcommand{\llmselect}{\textsf{LLMSelector}\xspace}

\newcommand{\rand}{\textsf{Random}\xspace}
\newcommand{\safeopt}{\textsf{SafeOpt}\xspace}

\newcommand{\budget}{\Lambda}

\newcommand{\bx}{\bm{x}}
\newcommand{\by}{\bm{y}}

\newcommand{\bI}{\bm{I}}

\newcommand{\bK}{\bm{K}}

\newcommand{\bX}{\bm{X}}

\newcommand{\cE}{\mathcal{E}}
\newcommand{\cF}{\mathcal{F}}

\newcommand{\cH}{\mathcal{H}}
\newcommand{\cI}{\mathcal{I}}
\newcommand{\cJ}{\mathcal{J}}

\newcommand{\cM}{\mathcal{M}}
\newcommand{\cN}{\mathcal{N}}

\newcommand{\cQ}{\mathcal{Q}}

\newcommand{\cS}{{\mathcal{S}}}
\newcommand{\cT}{{\mathcal{T}}}

\newcommand{\EE}{\mathbb{E}}

\newcommand{\NN}{\mathbb{N}}

\newcommand{\RR}{\mathbb{R}}

\newcommand{\btheta}{\bm{\theta}}

\newcommand{\argmin}{\mathop{\mathrm{argmin}}}

\newcommand{\KL}{{\rm KL}}

\newenvironment{customlegend}[1][]{%
    \begingroup
    \csname pgfplots@init@cleared@structures\endcsname
    \pgfplotsset{#1}%
}{%
    \csname pgfplots@createlegend\endcsname
    \endgroup
}%

\def\addlegendimage{\csname pgfplots@addlegendimage\endcsname}

\makeatletter
\newcommand\footnoteref[1]{\protected@xdef\@thefnmark{\ref{#1}}\@footnotemark}
\makeatother

\definecolor{RedPPAlt}{HTML}{DC143C}
\definecolor{Red}{HTML}{E81123}
\definecolor{Orange}{HTML}{FFB000}
\definecolor{Green}{HTML}{009E49}
\definecolor{LightBlue}{HTML}{00BCF2}
\definecolor{DeepBlue}{HTML}{001BA3}
\definecolor{Pink}{HTML}{7C3AED}
\definecolor{LightGreen}{HTML}{8BC63E}
\definecolor{Gray}{HTML}{6B7280}

\newcounter{myremark}
\renewcommand{\themyremark}{\arabic{myremark}}

\newtcolorbox{remark}[1][]{
    enhanced,
    breakable,
    before skip=10pt,
    after skip=10pt,
    colback=gray!5,
    colframe=navyblue!70!black,
    boxrule=0.3pt,
    leftrule=3pt,
    left=8pt,
    right=5pt,
    top=5pt,
    bottom=5pt,
    sharp corners,
    fonttitle=\bfseries,
    coltitle=black,
    title={Remark \refstepcounter{myremark}\themyremark #1},
    attach title to upper,
    after title={.\hspace{0.5em}},
}
\definecolor{navyblue}{RGB}{0, 0, 128}
\newcommand{\trref}[1]{Appendix~{#1}}

\AtBeginDocument{%
  }

\copyrightyear{2026}
\acmYear{2026}
\setcopyright{cc}
\setcctype{by}
\acmConference[KDD '26]{Proceedings of the 32nd ACM SIGKDD Conference on Knowledge Discovery and Data Mining V.2}{August 09--13, 2026}{Jeju Island, Republic of Korea}
\acmBooktitle{Proceedings of the 32nd ACM SIGKDD Conference on Knowledge Discovery and Data Mining V.2 (KDD '26), August 09--13, 2026, Jeju Island, Republic of Korea}
\acmDOI{10.1145/3770855.3818067}
\acmISBN{979-8-4007-2259-2/2026/08}

\begin{document}

\title{\textsc{SCOPE}: Cost-Efficient Model Selection for Compound AI Systems under Quality Constraints}
\subtitle{[Technical Report]}

\author{Yiqian Huang}
\orcid{0009-0001-5601-3439}
\affiliation{%
  \institution{National University of Singapore}
  \country{Singapore}
}
\email{yiqian@comp.nus.edu.sg}

\author{Shiqi Zhang}
\authornote{Corresponding author.}
\orcid{0000-0002-7155-9579}
\affiliation{%
  \institution{National University of Singapore}
  \country{Singapore}
}
\affiliation{%
  \institution{PyroWis AI}
  \country{Singapore}
}
\email{shiqi@pyrowis.ai}

\author{Tianyuan Jin}
\orcid{0009-0006-0266-5788}
\affiliation{%
  \institution{The Hong Kong University of Science and Technology (Guangzhou)}
    \country{China}
}
\email{tianyuan1044@gmail.com}

\author{Xiaokui Xiao}
\orcid{0000-0003-0914-4580}
\affiliation{%
  \institution{National University of Singapore}
    \country{Singapore}
}
\email{xkxiao@nus.edu.sg}

\begin{abstract}
A compound AI system consists of multiple LLM modules, together handling complex and multi-step tasks that exceed the capabilities of a single model. Existing systems often use a single expensive LLM across all modules to improve the result quality of the whole system. However, this configuration incurs prohibitive costs, particularly for data management and analytics tasks at scale, such as data manipulation.
To this end, we formalize the problem of \textit{constrained LLM selection for compound AI systems}, leveraging the diverse pricing and capabilities of different LLMs to achieve competitive quality at lower cost. Given a query dataset and a user-specified quality threshold, we aim to select an LLM for each module to minimize the system's average cost while ensuring that overall quality meets the required threshold.
To solve this problem, we propose \textsf{SCOPE}, a cost-efficient optimization algorithm. Unlike existing approaches that rely on expensive dataset-level evaluations, \textsf{SCOPE} exploits per-query results to rapidly estimate the system's cost and quality, and constructs confidence bounds to guide the search for promising LLM combinations. 
Furthermore, \textsf{SCOPE} provides theoretical guarantees for meeting the quality threshold and achieving near-optimal average cost.
We evaluate \textsf{SCOPE} against 7 baselines on three data processing tasks, demonstrating that it outperforms all baselines. Under the same search budget and quality constraint, it finds solutions with up to $20\times$ lower cost than the best competitor during the search and achieves up to $6\times$ lower final cost in the returned solution.
\end{abstract}

\maketitle
\newcommand{\ty}[1]{{\color{blue}[tianyuan: #1]}}
\section{Introduction}
A compound AI system is a multi-module system in which each module is implemented by a specific agent powered by a large language model (LLM). The system as a whole processes complex and multi-step tasks beyond the capability of any single model, thus finding applications in many scenarios, such as question answering~\cite{10.1145/3711896.3737012}, software development~\cite{hong2024metagpt}, and data processing~\cite{pourreza2023din, qian2024unidm}.
To ensure high result quality, existing compound AI systems typically use a state-of-the-art LLM for every module and emphasize frontier system design. However, more advanced LLMs usually incur higher inference costs. This monetary overhead becomes prohibitive for tasks that require digesting large volumes of data, which are common in the data management and analytics area~\cite{10.5555/3666122.3667957}. To exemplify, in data imputation, which fills missing entries in a table based on observed values, employing the flagship \texttt{GPT-5.2}~\cite{openai_gpt5.2_system_card_2025} across all modules of UniDM~\cite{qian2024unidm} costs around \$$180,000$ on a dataset with one million rows, each containing ten entities with rich semantic value spaces.

To mitigate this immense overhead, we formalize a new problem, {\it \ourproblem}, which shifts the focus from quality alone to both quality and cost. Given an $N$-module compound AI system, a candidate LLM set, a query dataset, and a quality threshold, this problem aims to assign a specific LLM to each module to minimize the system's average cost over the dataset while ensuring that the average quality meets the threshold. This formulation exploits the cost--quality Pareto frontier observed in modern LLMs~\cite{10.1145/3772429.3772445}, where differences in model pricing and capabilities create a rich search space for optimization. Importantly, \ourproblem admits solutions that achieve high quality at significantly lower cost, for example, by aligning model capabilities with the varying complexity of module-specific tasks.

Solving this problem remains challenging, since one must carefully balance cost and quality, using only limited observations of the system's performance on the given dataset. A natural approach is to treat the system as a black box and apply generic constrained optimization methods, such as Bayesian optimization~\cite{wang2025convergence,10.24963/ijcai.2023/486} and bandit algorithms~\cite{xu2023constrained,pmlr-v37-sui15}. Specifically, these solutions search over possible LLM assignments without knowing the system's internal structure, yet they fail to guarantee convergence to a near-optimal assignment that satisfies the quality constraint at minimum cost.
Moreover, they require a prohibitive number of trials, each of which executes the system on the \textit{entire dataset}, often incurring search costs that even exceed the potential savings.
Recently, compound-AI-specific methods~\cite{russo2025abacuscostbasedoptimizersemantic, chen2025optimizingmodelselectioncompound} attempt to accelerate the search by making strong structural assumptions. For instance, \abacus~\cite{russo2025abacuscostbasedoptimizersemantic} assumes independence between modules to estimate total cost and quality, while \llmselect~\cite{chen2025optimizingmodelselectioncompound} relies on monotonicity, assuming that upgrading a module's model always improves global quality. Unfortunately, these assumptions rarely hold in practice. Relying on ad-hoc heuristics without theoretical guarantees, these methods often fail to strictly satisfy the quality threshold.

To bridge this gap, we propose \ours, a cost-efficient optimization algorithm with rigorous theoretical guarantees. In contrast to existing methods that rely on coarse-grained dataset-level evaluations, \ours uses a query-level search strategy. Specifically, it treats individual query evaluations as atomic observations to iteratively tighten statistically valid confidence bounds on both cost and quality. This query-level granularity allows the algorithm to quickly prune unpromising LLM assignments without full-dataset evaluations, while the confidence bounds guide the search toward near-optimal solutions in a principled way.
For the underlying noisy constrained optimization problem, \ours guarantees that the returned solution simultaneously (i) satisfies the quality threshold with high probability and (ii) has a cost that converges to the optimum as the search budget increases.
In experiments, we extensively evaluate \ours against seven competitors across three distinct data processing tasks involving compound AI systems. Our evaluation spans a massive search space of tens of candidate models and millions of potential LLM assignments. Notably, \ours strictly satisfies the quality constraint while reducing average costs by over 20$\times$ compared to the best competing baseline. Even at the maximum budget, it achieves a 6$\times$ cost reduction. Furthermore, on test-time datasets, \ours demonstrates robust generalization, attaining an average cost of 2--5\% of a high-quality reference configuration while maintaining competitive average quality.

To summarize, we make the following contributions in this work:
\begin{itemize}[topsep=2pt,itemsep=1pt,parsep=0pt,partopsep=0pt,leftmargin=11pt]
\item We formalize \ourproblem as a constrained optimization problem (Section~\ref{sec:prelim}).
\item We propose \ours, a cost-efficient optimization algorithm that reduces the search cost (Section~\ref{sec:scope-general}).
\item We provide theoretical guarantees on both the cost and the quality of the solution returned by \ours (Section~\ref{sec:theory}).
\item We conduct experiments showing that \ours outperforms existing methods (Section~\ref{sec:exp}).
\end{itemize}

\section{Preliminaries}\label{sec:prelim}

\subsection{Problem Formulation}

\stitle{Compound AI system}
Consider a query dataset $\mathcal{Q}$ (where $|\mathcal{Q}|=Q$), a set of candidate LLMs $\cM$ (where $|\cM|=M$), and a compound AI system consisting of $N$ modules, each integrated with an LLM. A configuration is defined as a vector $\btheta = [\theta_1, \dots, \theta_N]$, where each $\theta_i \in \cM$ denotes the LLM selected for the $i$-th module. The configuration space is denoted by $\Theta = \cM^N$, representing the set of all possible configurations. For a configuration $\btheta \in \Theta$ and a query $q \in \mathcal{Q}$, system execution is measured by two metrics:
\begin{itemize}[topsep=2pt,itemsep=1pt,parsep=0pt,partopsep=0pt,leftmargin=11pt]
\item $\ell_s(\btheta, q) \in [0,1]$, the expected quality (e.g., accuracy) of the system's output on query $q$ under configuration $\btheta$;
\item $\ell_c(\btheta, q) \in [C_{\min}, C_{\max}]$, the corresponding expected monetary cost, bounded by known limits $C_{\min}$ and $C_{\max}$.
\end{itemize}
We focus on the {\it average quality} and {\it average cost} over the dataset $\mathcal{Q}$, which are defined, respectively, as
\begin{equation*}\label{equ:sum-cost-perf}
c(\btheta) = \frac{1}{Q}\sum_{q\in \mathcal Q}\ell_c(\btheta,q), \quad s(\btheta) = \frac{1}{Q}\sum_{q\in \mathcal Q}\ell_s(\btheta,q).
\end{equation*}

In practice, users typically begin with a {\it reference configuration} $\btheta_0$ (e.g., employing a powerful LLM for all modules), which yields high quality $s(\btheta_0)$ but incurs a high cost. Given $\epsilon \in (0,1)$ specifying the maximum allowable quality degradation, we define the {\it quality threshold} as $s_0 = (1-\epsilon)\cdot s(\btheta_0)$. Let $g(\btheta)=s_0-s(\btheta)$ and $\ell_g(\btheta, q)=s_0-\ell_s(\btheta, q)$ for each $q\in \cQ$.
We define \ourproblem as follows.

\begin{problem}[Constrained LLM Selection for Compound AI Systems]\label{problem-main}
Given an $N$-module compound AI system, a dataset $\mathcal Q$, a configuration space $\Theta=\cM^N$, and a reference configuration $\btheta_0 \in \Theta$ with known quality $s(\btheta_0)$, the objective is to find a configuration that minimizes the average cost while satisfying the quality constraint:
\begin{align}\label{eq:def:theta*}
    \btheta^\star \in \argmin_{\btheta \in \Theta: ~g(\btheta) \leq 0} c(\btheta).
\end{align}
\end{problem}

A configuration $\btheta \in \Theta$ is called {\it feasible} if it satisfies $g(\btheta) \leq 0$. By definition, $\btheta^\star$ exists since at least one configuration $\btheta_0$ is feasible.

\stitle{Notations}
Throughout this paper, we denote $[n]=\{1,\dots, n\}$ and $\NN=\{0, 1,\dots\}$. A boldface lowercase letter $\bx$ represents a vector, and a boldface uppercase letter $\bX$ represents a matrix, where $(\bX)_{a,b}$ denotes the value in its $a$-th row and $b$-th column. An algorithm for Problem~\ref{problem-main} runs over discrete time steps $t \in \mathbb{N}$. At each time $t$, it selects a pair $(\btheta_t, q_t) \in \Theta\times \cQ$ to evaluate (hereafter, to {\it observe}) and receives a pair of noisy observations of $\ell_c(\btheta_t, q_t)$ and $\ell_g(\btheta_t, q_t)$, denoted by $(y_{c,t}, y_{g,t})$.
Table~\ref{tab:notation} summarizes frequently used notations.

\begin{table}[!t]
\centering
\renewcommand{\arraystretch}{1.1}
\begin{small}
\caption{Frequently used notations.}\vspace{-2.4mm} \label{tab:notation}
\begin{tabular}{rp{2.2in}}	
    \toprule
    \bf Notation & \bf Description \\
    \midrule
    $N$   &  The number of modules in the system.\\
    $\cM, M$   &  The candidate LLM set with $|\cM|=M$.\\
    $\mathcal{Q}, Q$   &  The query dataset with $|\mathcal{Q}|=Q$.\\
    $\btheta, \Theta$   &  A configuration $\btheta=[\theta_1,\dots,\theta_N]$ from $\Theta=\cM^N$.\\
    $\ell_s(\btheta,q)$   &  Query-wise quality of $\btheta$ and $q$.\\
    $\ell_c(\btheta,q)$   &  Query-wise cost of $\btheta$ and $q$.\\
    $s(\btheta), c(\btheta)$   &  Average quality and average cost of $\btheta$.\\
    $g(\btheta)$   &  Constraint function $(=s_0-s(\btheta))$.\\
    $\ell_g(\btheta, q)$   &  Query-wise constraint function $(=s_0-\ell_s(\btheta, q))$.\\
    $y_{c,t}, y_{g,t}$ & The observations made by an algorithm at time $t$.\\
    
    \bottomrule
\end{tabular}
\end{small}
\end{table}

\subsection{Existing Solutions}\label{sec:generic}
Problem~\ref{problem-main} is challenging because both $c$ and $g$ are expensive black-box functions that do not admit closed-form expressions.
A natural baseline is to cast it as a generic constrained optimization problem, where one has access only to \textit{noisy observations} (i.e., values with random noise) of $c(\btheta)$ and $g(\btheta)$. Under this framing, one could apply sequential search methods for constrained optimization, such as constrained Bayesian optimization, among others~\cite{xu2023constrained,pmlr-v37-sui15,wang2025convergence,10.24963/ijcai.2023/486,10.5555/3600270.3600272}.

\stitle{Generic optimization approaches}
Generic optimization methods, such as Bayesian optimization (BO)~\cite{wang2025convergence,10.24963/ijcai.2023/486} and bandit algorithms~\cite{xu2023constrained,pmlr-v37-sui15,10.5555/3600270.3600272}, have been used to tune compound AI systems for cost or quality~\cite{opsahl-ong-etal-2024-optimizing,khattab2024dspy, he2025cognify}, and can be rigorously extended to handle the quality constraint. In the unconstrained setting, these methods come with theoretical guarantees in terms of \textit{regret}, which quantify how quickly the returned configuration approaches the optimal configuration.
Such guarantees typically rely on standard assumptions: informally, \textit{similar configurations yield similar function values}, so that past evaluations can generalize to unobserved configurations.

\stitle{Compound-AI-specific approaches}
Other recent methods aim to improve empirical \eat{efficiency}search efficiency by making extra assumptions about the compound AI system, rather than treating it as a pure black box.
For example, \abacus~\cite{russo2025abacuscostbasedoptimizersemantic}, the optimizer used by Palimpzest~\cite{palimpzestCIDR}, estimates $c(\cdot)$ and $s(\cdot)$ by assuming independence across modules and employing additive or multiplicative estimators.
\llmselect~\cite{chen2025optimizingmodelselectioncompound} assumes that quality is monotonically non-decreasing when upgrading LLMs according to a fixed ranking, and accordingly proposes a round-robin procedure that updates one module at a time. \llambo~\cite{liu2024large} treats the LLM itself as an estimation model, using its internal knowledge and the observation history to predict the cost and quality of unseen configurations.

\stitle{Limitations}
Existing approaches have two main limitations in the noisy constrained setting~\cite{xu2023constrained, wang2025convergence}.  First, it is hard for them to guarantee \textit{correctness} and \textit{effectiveness} at the same time. 
Here, correctness means the returned solution $\btheta_{\rm out}$ is feasible, i.e., $g(\btheta_{\rm out})\le 0$, and effectiveness means its cost is close to the best feasible value, i.e.,
$c(\btheta_{\rm out}) \approx c(\btheta^\star)$.
For instance, methods such as \config~\cite{xu2023constrained} prioritize effectiveness but risk violating correctness, while \safeopt~\cite{pmlr-v37-sui15} ensures correctness but often converges to suboptimal solutions. While recent work on {\it constrained expected improvement} (\BO)~\cite{wang2025convergence}, a variant of constrained BO, theoretically addresses both, it relies on a \textit{noiseless} assumption (i.e., without random noise) that does not hold for stochastic LLM outputs.
Second, existing approaches remain {\it cost-inefficient} as they do not fully exploit query-level signals $\ell_c$ and $\ell_g$. By treating average cost and average quality as pure black-box functions, these methods must evaluate configurations on the entire dataset $\cQ$ to obtain an observed signal for average cost or quality. As shown in Section~\ref{sec:exp}, fully evaluating an infeasible configuration with $Q=500$ can cost around $\$12$, whereas analyzing just a few queries using $\ell_c$ and $\ell_g$ could reveal infeasibility immediately. Yet, it remains an open problem to exploit query-level signals without compromising theoretical guarantees. These limitations motivate a principled solution that addresses Problem~\ref{problem-main} by integrating rigorous guarantees with a cost-efficient, query-level evaluation scheme.

\section{\ours}\label{sec:scope-general}
We propose \textbf{S}equential \textbf{C}onfidence-bound-based \textbf{O}ptimization via \textbf{P}artial \textbf{E}valuation (\ours), which solves Problem~\ref{problem-main} with guarantees of correctness and effectiveness. When $Q=1$, \ours also applies to generic noisy constrained optimization.

\begin{algorithm}[t]
\caption{\ours $(\budget, \delta, \alpha, \theta_{\rm base})$}
\label{alg:scope-main}
\KwIn{Budget $\budget > 0$, failure probability $\delta \in (0,1)$, parameter $\alpha\in (0, 1/2)$, base model $\theta_{\rm base}\in \cM$}
\KwOut{Returned solution $\btheta_{\mathrm{out}}$, stopping time $\tau$}
$t_0, \{\btheta_i, q_i, y_{c,i}, y_{g,i}\}_{i=1}^{t_0}\gets{\warmstart}(\theta_{\rm base})$ (Algorithm~\ref{alg:bound-calibration})\label{alg:scope-main:init0}\;
Define $L_{\zeta, t}, U_{\zeta, t}\gets$~\eqref{eq:lcb-ours} for $t\in \mathbb{N}, \zeta\in\{c, g\}$\label{alg:scope-main:rules}\;
$t\gets t_0,~\btheta_{\rm out}\gets \btheta_0,~U_{\rm out}\gets U_{c,t_0}(\btheta_0)$\label{alg:scope-main:init1}\;
\For{$i=1,2,\dots$\label{alg:scope-main:while}}{
${\displaystyle \btheta_{\text{cand}} \gets \argmin_{\btheta\in\Theta: L_{g,t}(\btheta)\leq -i^{-\alpha}} L_{c,t}(\btheta)}$\;\label{alg:scope-main:select}
Sort $\cQ$ as $\{q^{(1)}, \dots, q^{(Q)}\}$ by decreasing $\varphi_i(q)$ in~\eqref{eq:uncertainty-score}\label{alg:scope-main:sort}\;
\For{$j=1,\dots, Q$\label{alg:scope-main:for}}{
    $t\gets t+1$\;
    Observe at $(\btheta_{\text{cand}}, q^{(j)})$ to obtain $(y_{c,t}, y_{g, t})$\label{alg:scope-main:observe}\;
\If{$U_{c,t}(\btheta_{\rm cand}) \le U_{\rm out}$ $\mathbf{and}$ $\min\{U_{g,t}(\btheta_{\rm cand}), U_{g,t-1}(\btheta_{\rm cand})\} \le 0$\label{alg:scope-main:update}}{
    $U_{\rm out}\gets U_{c,t}(\btheta_{\rm cand})$\;
    $\btheta_{\rm out}\gets \btheta_{\rm cand}$\;
}
\lIf{$\sum_{t'=1}^{t}y_{c,t'}>\budget$\label{alg:scope-main:stop}}{\textbf{break all loops}}
\lIf{$L_{g,t}(\btheta_{\text{cand}}) > 0~\mathbf{or}$
     $L_{c,t}(\btheta_{\text{cand}}) > U_{\rm out}$\label{alg:scope-main:break}}{\textbf{break the $j$-loop}}
}
}
\Return{$(\btheta_{\mathrm{out}}, \tau\gets t)$}

\end{algorithm}

\subsection{Main Procedure}

At the core of \ours are carefully designed {\it confidence bounds} for the cost $c(\cdot)$ and the constraint $g(\cdot)$. Given a failure probability $\delta\in(0,1)$, for time $t\in\NN$, configuration $\btheta\in\Theta$, and function $\zeta\in\{c,g\}$, \ours maintains lower and upper bounds $L_{\zeta,t}(\btheta)$ and $U_{\zeta,t}(\btheta)$ such that they enclose the true function values with probability at least $1-\delta$: 
\begin{equation}\label{eq:lcb-ucb-property}
\Pr\left[\forall \zeta, t, \btheta:\
L_{\zeta,t}(\btheta) \le \zeta(\btheta) \le U_{\zeta,t}(\btheta)
\right] \ge 1-\delta.
\end{equation}
Building on these bounds, \ours proceeds in iterations. In each iteration $i\in\{1,2,\dots\}$, \ours selects a candidate configuration $\btheta_{\rm cand}$ that is promising under the current bounds and then evaluates $\btheta_{\rm cand}$ sequentially on queries in $\cQ$ until the search budget is exhausted. Notably, the guarantees of \ours are ensured by two key mechanisms.
\begin{enumerate}
[topsep=2pt,itemsep=1pt,parsep=0pt,partopsep=0pt,leftmargin=14pt]
    \item \textbf{Candidate selection:} In iteration $i$, it selects $\btheta_{\rm cand}$ using a dynamic constraint $L_{g,t}(\btheta)\le -i^{-\alpha}$, where $\alpha\in(0,1/2)$ is a fixed parameter, to filter out candidates that are likely infeasible. Among the remaining candidates, $\btheta_{\rm cand}$ is chosen greedily by minimizing the cost lower bound $L_{c,t}(\btheta)$.
    \item \textbf{Sequential query evaluation:} After selecting $\btheta_{\rm cand}$, \ours evaluates it on queries $q^{(1)},\dots,q^{(Q)}\in\cQ$ in decreasing order of a score $\varphi_i(q)$. This score is defined to quantify how much information we gain from evaluating $q$ and, in turn, to tighten the bounds for both $c$ and $g$ using as few observations as possible.
\end{enumerate}
In addition, the cost efficiency of \ours hinges on two further design choices that exploit query-level signals.
\begin{enumerate}
[start=3,topsep=2pt,itemsep=1pt,parsep=0pt,partopsep=0pt,leftmargin=14pt]
\item \textbf{Bound calibration:} During initialization, \ours invokes a subroutine, \warmstart, which calibrates accurate confidence bounds within a small initial budget.
\item \textbf{Pruning mechanism:} \ours incorporates a pruning mechanism enabled by the sequential, query-level evaluation. As the bounds are updated, it stops evaluating $\btheta_{\rm cand}$ early once $\btheta_{\rm cand}$ appears infeasible (i.e., $L_{g,t}(\btheta_{\rm cand})>0$) or cannot improve on the current best solution in terms of cost.
\end{enumerate}

Algorithm~\ref{alg:scope-main} gives the pseudocode of \ours; we defer the description of \warmstart to Section~\ref{sec:calibration} and defer the definition of $L_{\zeta,t}$, $U_{\zeta,t}$, and $\varphi_i$ to Section~\ref{sec:confidence-bounds}.
It takes four additional inputs: a {search budget} $\budget>0$ which limits the cumulative monetary cost of observing $y_{c,t}$ over time, a failure probability $\delta\in(0,1)$ used in the confidence bounds, a trade-off parameter $\alpha\in(0,1/2)$, and a {base model} $\theta_{\rm base}\in\cM$, which represents a user-chosen cost-efficient LLM used by the \warmstart subroutine.
In Line~\ref{alg:scope-main:init0}, the algorithm first invokes \warmstart with parameter $\theta_{\rm base}$, which performs $t_0$ initial observations, and sets $t\gets t_0$. In Lines~\ref{alg:scope-main:rules}--\ref{alg:scope-main:init1}, the algorithm initializes the confidence bounds $L_{\zeta,t}$ and $U_{\zeta,t}$, the current best configuration $\btheta_{\mathrm{out}} \gets \btheta_0$, and the current best upper bound $U_{\rm out}$.
After initialization, at each outer loop $i$ (Line~\ref{alg:scope-main:while}), \ours selects a configuration $\btheta_{\mathrm{cand}}$ according to the selection criterion in Line~\ref{alg:scope-main:select}. 
Subsequently, the queries in $\cQ$ are ordered as $q^{(1)},\dots,q^{(Q)}$ in decreasing order of the query-wise surrogate uncertainty $\varphi_i(q^{(j)})$ (Line~\ref{alg:scope-main:sort}), and are sequentially evaluated with $\btheta_{\rm cand}$ in the inner $j$-loop (Line~\ref{alg:scope-main:for}). At each inner iteration $j$, the algorithm increments the time $t$ by one and observes $(\btheta_{\rm cand}, q^{(j)})$ to obtain $y_{c,t}$ and $y_{g,t}$ (Line~\ref{alg:scope-main:observe}), which in turn refines the confidence bounds at time $t$. In Line~\ref{alg:scope-main:update}, if $U_{c,t}(\btheta_{\rm cand})\le U_{\rm out}$ and $\min\{U_{g,t}(\btheta_{\rm cand}), U_{g,t-1}(\btheta_{\rm cand})\}\le 0$ hold (i.e., $\btheta_{\rm cand}$ is certified feasible with high probability at either time), the algorithm updates $U_{\rm out}$ and $\btheta_{\rm out}$. Additionally, if the search budget is exhausted, the algorithm breaks out of all loops and terminates (Line~\ref{alg:scope-main:stop}). Otherwise, $\btheta_{\rm cand}$ is subject to early pruning where \ours breaks the $j$-loop (Line~\ref{alg:scope-main:break}) if $L_{g,t}(\btheta_{\rm cand}) > 0$ or $L_{c,t}(\btheta_{\rm cand}) > U_{\rm out}$.
Finally, \ours returns $(\btheta_{\rm out},\tau)$ where $\tau=t$ represents the stopping time after which it makes no further observations.

\subsection{Bound Calibration}\label{sec:calibration}

\begin{algorithm}[t]
\caption{\warmstart$(\theta_{\rm base})$}
\label{alg:bound-calibration}
\KwIn{Base model $\theta_{\rm base}\in\cM$}
\KwOut{Time step $t_0$, history $\{(\btheta_t,q_t,y_{c,t},y_{g,t})\}_{t=1}^{t_0}$}

$t_0 \gets 0,~\Theta_{\text{init}} \gets$~\eqref{eq:theta-1-def}, $\cQ_{0}\gets \emptyset$\label{alg:scope-main:bc-sample-theta}\;
Sort $\cQ$ as $\{q^{(1)},\dots,q^{(Q)}\}$ in a random order\;
\For{$j=1,\dots,\lceil\log_2 Q\rceil+1$}{\label{alg:scope-main:bc-for}
    $\cQ_{j} \gets\{q^{(1)},\dots, q^{(\min\{2^{j-1},Q\})}\}$\;\label{alg:scope-main:bc-qprefix}
    \ForEach{$\btheta, q \in \Theta_{\text{init}}\times (\cQ_{j}\setminus \cQ_{j-1})$}{\label{alg:scope-main:bc-foreach}
            $t_0\gets t_0+1$\;
            Observe at $(\btheta,q)$ to obtain $(y_{c,t_0}, y_{g,t_0})$\;
    }
    $\Theta_{\text{init}} \gets$ top-$\lceil{|\Theta_{\text{init}}|}/{2}\rceil$ configurations in $\Theta_{\text{init}}$
    ranked by $S(\btheta)$, where $S(\btheta)=-\sum_{t\in[t_0]:\,\btheta_t=\btheta,~q_t\in \cQ_{j}} y_{g,t}$\;\label{alg:scope-main:bc-halving}
}\label{alg:scope-main:bc-end}
\Return{$(t_0,\;\{(\btheta_t,q_t,y_{c,t},y_{g,t})\}_{t=1}^{t_0})$}
\end{algorithm}

To identify the rough scale of confidence bounds for $c(\cdot)$ and $g(\cdot)$, \ours invokes a heuristically designed \warmstart subroutine. \warmstart proceeds in a hybrid manner: it first collects observations from a large, diverse pool of configurations using a single query, and then iteratively doubles the size of the query set while halving the configuration pool. As a result, higher-quality configurations receive additional evaluation, until a single remaining configuration in the pool is evaluated on the entire query set.

This design is motivated by three empirical observations. First, cost differences are mainly driven by the selected LLMs rather than by which query is executed, since model prices often differ by orders of magnitude. As a result, evaluating many diverse configurations on only a few queries provides early estimates of the cost scale across configurations.
Second, evaluating a high-potential configuration on only a few queries is often insufficient, because the query-wise quality $\ell_s(\btheta,q)$ can vary significantly across queries $q\in\cQ$. Given a limited search budget $\budget$, we increase the number of sampled queries only for configurations that appear high-quality, to obtain a reliable estimate of their average quality. Finally, it is often beneficial to choose a lightweight model $\theta_{\rm base}$ and focus on configurations related to this model, thereby saving search cost by avoiding evaluations of configurations that include many expensive models.

The detailed pseudocode is shown in Algorithm~\ref{alg:bound-calibration}, which takes as input the base model $\theta_{\rm base}\in\cM$. The subroutine first constructs an initial pool $\Theta_{\text{init}}\subseteq\Theta$ as
\begin{equation}\label{eq:theta-1-def}
\textstyle
\Theta_{\rm init}
=\big\{[\theta_1,\dots,\theta_N]\in\Theta:\ 
\sum_{i=1}^{N}\bm 1\{\theta_i\neq \theta_{\text{base}}\}\le 1
\big\},
\end{equation}
and randomly shuffles $\cQ$ as $\{q^{(1)},\dots,q^{(Q)}\}$. Then, it runs for $\lceil\log_2(Q+1)\rceil$ rounds. In round $j$, it evaluates each remaining configuration in $\Theta_{\text{init}}$ on additional queries in $\cQ_j\setminus \cQ_{j-1}$, where $\cQ_j=\{q^{(1)},\dots,q^{(\min\{2^{j-1},Q\})}\}$. Afterward, it retains the top half of configurations ranked by their cumulative observed quality over $\cQ_j$ and discards the rest. By the final round, the pool is reduced to a single configuration, which has been evaluated on all $Q$ queries.

\subsection{Confidence Bounds}\label{sec:confidence-bounds}
For each $\zeta\in\{c,g\}$, $t\in \NN$, and $\btheta\in\Theta$, the confidence bounds in Line~\ref{alg:scope-main:rules} are of the following form:
\begin{equation}\label{eq:lcb-ours}
\begin{aligned}
L_{\zeta,t}(\btheta)&= \bar\mu_{\zeta,t}(\btheta) - \beta_{\zeta,t}\,\bar\sigma_{\zeta,t}(\btheta),\\
U_{\zeta,t}(\btheta)&= \bar\mu_{\zeta,t}(\btheta) + \beta_{\zeta,t}\,\bar\sigma_{\zeta,t}(\btheta).
\end{aligned}
\end{equation}
Here, $\bar\mu_{\zeta,t}(\btheta)$ is the {\it surrogate mean} estimating $\zeta(\btheta)$, $\bar\sigma_{\zeta,t}(\btheta)$ is the {\it surrogate standard deviation} representing the uncertainty in this estimate, and the coefficient $\beta_{\zeta,t}$ scales this uncertainty to ensure valid confidence bounds.
To instantiate \eqref{eq:lcb-ours}, we fix hyperparameters $R_c, R_g \ge 0$, $B_c, B_g > 0$, and a symmetric positive definite kernel (hereafter, {\it SPD kernel}) function $k:\Theta\times\Theta\to\mathbb{R}$ with $k(\btheta,\btheta)=1$ for all $\btheta\in\Theta$. Note that these hyperparameters are used in Assumptions~\ref{ass:observation}--\ref{ass:rkhs} (see Section~\ref{sec:assumptions}), under which we derive the guarantees.
$\bar\mu_{\zeta,t}(\btheta)$ and $\bar\sigma_{\zeta,t}(\btheta)$ are defined based on standard zero-mean Gaussian process (GP) regression with kernel $k$, as described below.

\begin{definition}[Zero-mean GP regression~\cite{rasmussen2003gaussian}]\label{def:gp}
Fix a kernel $k:\Theta\times \Theta\rightarrow\RR$ and a constant $\lambda > 0$. Given $J$ observations $\{(\btheta^{(i)},y^{(i)})\}_{i=1}^{J}$ in an arbitrary order, where $y^{(i)}$ is obtained by observing a particular function $f:\Theta\rightarrow \RR$ at $\btheta^{(i)}\in\Theta$, define
\begin{align*}
&\by=[y^{(1)}, \dots, y^{(J)}]^\top, \quad \bx=[\btheta^{(1)},\dots,\btheta^{(J)}],\\
&\bm k_{\bx}(\btheta)=[k(\btheta,\btheta^{(1)}),\dots,k(\btheta,\btheta^{(J)})]^\top,\\
&(\bm{K}_{\bx})_{a,b} = k(\btheta^{(a)}, \btheta^{(b)}), \quad a,b\in [J].
\end{align*}
Then, for any $\btheta\in \Theta$, the zero-mean GP regression defines
\begin{equation}\label{equ:gp-regression}
\begin{aligned}
 \hat\mu_{\bx,\by}(\btheta)
&= \bm k_{\bx}(\btheta)^\top(\bm K_{\bx}+\lambda \bm I)^{-1}\by,\\
(\hat\sigma_{\bx,\by}(\btheta))^2
&= k(\btheta,\btheta)-\bm k_{\bx}(\btheta)^\top(\bm K_{\bx}+\lambda \bm I)^{-1}\bm k_{\bx}(\btheta),
\end{aligned}
\end{equation}
In particular, if $\by=\emptyset$, then $\hat\mu_{\bx,\by}(\btheta)=0$ and $\hat\sigma_{\bx,\by}(\btheta)^2=k(\btheta,\btheta)$.
\end{definition}

For each $t\in \NN$ under \ours, define $\cJ_{q,t}=\{j\in[t]: q_j=q\}$, ordered as $j_q(1)<\cdots<j_q({|\cJ_{q,t}|})$, and denote
\begin{equation}\label{eq:j_max}
J_{q, t}=|\cJ_{q,t}|, \qquad J_{\max,t}\ =\ \max_{q\in\cQ}J_{q,t}.
\end{equation}
Let $\bx_{q,t}=[\btheta_{j_q(1)}, \dots, \btheta_{j_q({J_{q,t}})}], \;\by_{\zeta,q,t}=[y_{\zeta,j_q(1)},\dots,y_{\zeta,j_q({J_{q,t}})}]^\top$ if $J_{q,t}>0$, otherwise $\bx_{q,t}=\by_{\zeta,q,t}=\emptyset$. We then define the surrogate mean and standard deviation in \eqref{eq:lcb-ours} as
\begin{equation}\label{eq:scope-mean-variance}
\bar\mu_{\zeta,t}(\btheta)=\sum_{q\in\mathcal{Q}}\frac{\hat\mu_{\bx_{q,t},{\by}_{\zeta,q,t}}(\btheta)}{Q},\quad
\bar\sigma_{\zeta,t}(\btheta)^2=\sum_{q\in\mathcal{Q}}\left(\frac{\hat\sigma_{\bx_{q,t},{\by}_{\zeta,q,t}}(\btheta)}{Q}\right)^2,
\end{equation}
where we apply Definition~\ref{def:gp} to $f=\ell_\zeta(\cdot, q)$ with $\bx_{q,t}$ and $\by_{\zeta,q,t}$ for each $q\in \cQ$, and set $\lambda=\max\{R_{c}^2, R_{g}^2, 10^{-9}\}$ in Definition~\ref{def:gp} for simplicity.
The coefficient $\beta_{\zeta,t}$ is based on the {\it maximum information gain}~\cite{srinivas2009gaussian}, defined as
\begin{equation*}
\textstyle
\gamma(J) =
\max_{A\subseteq \Theta:\ |A|\leq J} \frac{1}{2}\log\det(\bI+\lambda^{-1}\bK_A),
\end{equation*}
where $\bK_A\in\RR^{J'\times J'}$ $(J'\leq J)$ is defined by $(\bK_A)_{i,j}=k(\btheta^{(i)},\btheta^{(j)})$, with $A=\{\btheta^{(1)},\dots,\btheta^{(J')}\}\subseteq \Theta$ in any order.
Fix $\delta\in (0,1)$.
Given $\gamma(J)$ and $\lambda$ above, we set $\beta_{\zeta,t}$ as
\begin{equation}\label{eq:beta-def}
\beta_{\zeta,t}\;=\;
\sqrt{Q}\cdot\left(B_\zeta
\;+\;
\frac{R_\zeta}{\sqrt\lambda}\sqrt{2\Big(\gamma(J_{\max,t})+\log\frac{2Q}{\delta}\Big)}\right),
\end{equation}
where $J_{\max,t}$ is defined in~\eqref{eq:j_max}.
Furthermore, in Line~\ref{alg:scope-main:sort} in iteration $i$ at time $t$, we define
\begin{equation}\label{eq:uncertainty-score}
\varphi_{i}(q)=\hat\sigma_{\bx_{q,t},{\by}_{c,q,t}}(\btheta_{\rm cand}),
\end{equation}
where $\hat\sigma_{\bx_{q,t},{\by}_{c,q,t}}$ is defined in~\eqref{eq:scope-mean-variance}. By definition, one can verify that it also equals $\hat\sigma_{\bx_{q,t},\by_{g,q,t}}$.

\section{Theoretical Analysis}\label{sec:theory}

\subsection{Assumptions}\label{sec:assumptions}
We state the theoretical assumptions under which \ours's guarantees are derived. They extend the standard {\it frequentist setting} for generic constrained optimization~\cite{JMLR:v12:bull11a, chowdhury2017kernelized, srinivas2009gaussian} to $Q \ge 1$. When $Q=1$, they reduce to the standard ones in~\cite{srinivas2009gaussian}. These assumptions are parameterized by $R_c, R_g \ge 0$, $B_c, B_g > 0$, and the SPD kernel $k$ with $k(\btheta,\btheta)=1$ for all $\btheta\in\Theta$.

\begin{assumption}[Sub-Gaussian noise]\label{ass:observation}
The observations satisfy
\begin{equation}\label{eq:sub-gaussian-noise}
\forall t\in \NN, \;\;
y_{c, t} = \ell_c(\btheta_t, q_t) + \eta_{c,t},
\quad
y_{g, t} = \ell_g(\btheta_t, q_t) + \eta_{g,t},
\end{equation}
where $\eta_{c,t}$ and $\eta_{g,t}$ are zero-mean random variables. For each $\zeta\in\{c,g\}$, the sequence $\{\eta_{\zeta,t}\}_{t=1}^{\infty}$ is {conditionally $R_\zeta$-sub-Gaussian}, i.e., 
$$
\forall t\in \NN, \forall u\in\RR,~\EE[\exp{(u\cdot\eta_{\zeta, t})} | \cF_{t-1}]\leq \exp{\big({(u^2 R_\zeta^2)} / 2\big)},
$$
where $\cF_{t-1}$ is the $\sigma$-algebra generated by $\{\btheta_i, q_i, y_{c,i},y_{g,i}\}_{i=1}^{t-1}$.
\end{assumption}

\begin{assumption}[Boundedness in RKHS]\label{ass:rkhs}
Let $\cH_k$ be the \emph{reproducing kernel Hilbert space} (RKHS) induced by $k$, with norm $\|\cdot\|_{\cH_k}$.
\footnote{Concretely, RKHS $\cH_k$ is a collection of functions equipped with an inner product $\langle\cdot,\cdot\rangle_{\cH_k}$ such that for every $\btheta\in\Theta$ and $f\in\cH_k$, it holds that $k(\btheta,\cdot)\in\cH_k$ and $\langle f, k(\btheta,\cdot)\rangle_{\cH_k}=f(\btheta)$. The norm is then $\|f\|_{\cH_k}=\sqrt{\langle f,f\rangle_{\cH_k}}$.} We assume for each $q\in \cQ$, $\|\ell_c(\cdot, q)\|_{\cH_k}\leq B_c$ and $\|\ell_g(\cdot, q)\|_{\cH_k}\leq B_g$.
\end{assumption}

\stitle{Interpretation}
Regarding Assumption~\ref{ass:observation}, upper bounds for $R_c$ and $R_g$ are often available in practice. For example, if system evaluations satisfy $y_{c,t}\in [C_{\min}, C_{\max}]$ and $y_{g,t}\in [s_0-1, s_0]$, a sufficient choice is $R_c = (C_{\max}-C_{\min})/2$ and $R_g = 1/2$. If evaluations are noiseless, it suffices to set any $R_c=R_g \geq 0$. In Assumption~\ref{ass:rkhs}, $B_{\zeta}$ bounds how similar $\ell_\zeta(\btheta, q)$ and $\ell_\zeta(\btheta', q)$ are for every $q\in\cQ$, $\zeta\in\{c,g\}$, and $\btheta, \btheta'\in\Theta$. In particular, one can infer the following property:
\[
|\ell_\zeta(\btheta, q)-\ell_\zeta(\btheta', q)|\leq B_\zeta \sqrt{2-2k(\btheta, \btheta')}.
\]
Since $\Theta$ is finite in Problem~\ref{problem-main}, for any SPD kernel $k$ there exist $B_c, B_g < \infty$ such that Assumption~\ref{ass:rkhs} holds. In implementation, this allows us to adopt a commonly used kernel and treat $B_c$ and $B_g$ as hyperparameters, following prior work~\cite{xu2023constrained, srinivas2009gaussian}.

\subsection{Guarantees}
Recall that $\tau$ is the stopping time after which the algorithm makes no further observations, and $\btheta_{\rm out}$ is the returned configuration. We say the algorithm is \textit{$\delta$-correct} if $\Pr[g(\btheta_{\rm out})\le 0]\ge 1-\delta$, where the randomness is over the observations. We measure effectiveness by \textit{simple regret}, defined as $\mathrm{SR}(\tau)=c(\btheta_{\rm out})-c(\btheta^\star)$, where $c(\btheta^\star)$ is the optimal cost defined in~\eqref{eq:def:theta*}.
The following theorem shows that the confidence bounds hold with high probability. {\it All proofs are provided in \trref{\ref{appendix:proof}}.}

\begin{theorem}\label{cla:bounds}
Fix $\delta\in (0,1)$. For any $t\in\NN$, $\zeta\in\{c,g\}$, and $\lambda>0$ which is used in Definition~\ref{def:gp}, let the confidence bounds in~\eqref{eq:lcb-ours} be defined using~\eqref{eq:beta-def} and \eqref{eq:scope-mean-variance}.
Under Assumptions~\ref{ass:observation}--\ref{ass:rkhs}, the resulting bounds satisfy the high-probability event in~\eqref{eq:lcb-ucb-property}.
\end{theorem}

The next theorem shows that \ours is $\delta$-correct and provides an explicit upper bound on $\mathrm{SR}(\tau)$.

\begin{theorem}\label{thm:main-T-ours}
Fix $\budget>0$, $\delta\in(0,1)$, and $\alpha\in(0,1/2)$, and let $t_0$ be the time returned by \warmstart in Line~\ref{alg:scope-main:init0}.
Under Assumptions~\ref{ass:observation}--\ref{ass:rkhs},
\ours is $\delta$-correct.
Moreover, for any $\tau$ returned by \ours, with probability at least $1-\delta$, the simple regret satisfies
\begin{equation}\label{eq:sr-bound-T}
\mathrm{SR}(\tau)
\ \le\
\frac{\mathsf{A}}{\sqrt{(\tau-t_0)\ -\ \mathsf{B}\cdot(\tau-t_0)^{2\alpha}\ -\ \mathsf{C}}}
\end{equation}
whenever $g(\btheta^\star)<0$ and $(\tau-t_0)- \mathsf{B}\cdot(\tau-t_0)^{2\alpha}- \mathsf{C}>0,$
where
\begin{equation*}\label{eq:sr-constants}
\begin{gathered}
\mathsf A=
8\beta_{c,\tau}\sqrt{Q(\lambda+1)\gamma(\tau)},\quad \mathsf B=16Q(\lambda+1)\beta_{g,\tau}^2\gamma(\tau),\\
\mathsf C=Q\lceil{(-g(\btheta^\star)})^{-1/\alpha}\rceil.
\end{gathered}
\end{equation*}
\end{theorem}

As shown in Theorem~\ref{thm:main-T-ours}, $\mathrm{SR}(\tau)$ in Algorithm~\ref{alg:scope-main} converges at a rate of $O(1/\sqrt{\tau})$, since $\Theta$ is finite and hence $\gamma(\tau)$ is bounded by a constant. 
The parameter $\alpha$ trades off the $(\tau-t_0)^{2\alpha}$ term against the dependence on $g(\btheta^\star)$ via $\mathsf C$. Since $g(\btheta^\star)$ is unknown, a convenient choice is to set $\alpha$ to a moderate constant, e.g., $1/3$. The dependence on $g(\btheta^\star) < 0$ is necessary and, as $\alpha \to 1/2$, the term $(-g(\btheta^\star))^{-2}$ becomes nearly tight, in the sense that it is unavoidable in the worst case for any $\delta$-correct algorithm; see \trref{\ref{appendix:impossibility}} for details. Moreover, the following corollary translates the $\tau$-based bound into an $O(\sqrt{(\log \budget)/\budget})$ bound in expectation.

\begin{corollary}\label{cor:exp-sr-budget}
Fix $\alpha\in(0,1/2)$ and suppose Assumptions~\ref{ass:observation}--\ref{ass:rkhs} and $g(\btheta^\star)<0$ hold.
For each search budget $\budget>0$, let $\delta= \budget^{-2}$.
Then there exists $m\ge 2$, independent of $\budget$, such that for all $\budget\ge m$,
\(
\EE[\mathrm{SR}(\tau)]
=
O(\sqrt{(\log \budget)/\budget}),
\)
where the expectation is over the randomness of the observations and the stopping time $\tau$.
\end{corollary}

\subsection{Computational Complexity}
While simple regret captures the main search cost, we also analyze the additional computational cost to assess the practical runtime of \ours.
In the algorithm, the time and space complexity are dominated by the selection of $\btheta_{\rm cand}$ in Line~\ref{alg:scope-main:select}. As implied by~\eqref{eq:scope-mean-variance}, this step requires computing $\hat\mu_{\bx_{q,t},{\by}_{\zeta,q,t}}$ and $\hat\sigma_{\bx_{q,t},{\by}_{\zeta,q,t}}$ for all $\btheta\in\Theta$, $\zeta\in\{c,g\}$, and $q\in\cQ$ at time $t$.
To compute each $\hat\mu_{\bullet}$ and $\hat\sigma_{\bullet}$, we use standard GP updates as detailed in \citet{chowdhury2017kernelized}, which costs $O(t^2)$ time and $O(t^2)$ space. Therefore, the time complexity of Line~\ref{alg:scope-main:select} at time $t$ is $O(|\Theta|\sum_{q\in\cQ} J_{q,t}^2)
= O(Q\cdot|\Theta|\cdot J_{\max,t}^2).$ Summing over $t\in [\tau]$, the worst-case total time complexity of \ours in terms of $\tau$ is
\[
 O\!\left(\sum_{t=1}^{\tau} Q\cdot|\Theta|\cdot J_{\max,t}^2\right)=O\!\left(Q\cdot|\Theta|\cdot \tau^3\right).
\]
This is larger than the $O\left(|\Theta|\cdot \tau^3\right)$ running time of existing BO and bandit approaches~\cite{wang2025convergence, xu2023constrained}, due to the query-wise exploitation in \ours.
In addition, computing $\beta_{\zeta, t}$ in~\eqref{eq:beta-def} also involves $\gamma(J)$ for all $J\leq \tau$. Since $\frac{1}{2}\log\det(\bI+\lambda^{-1}\bK_A)$ is a non-decreasing submodular set function over $A\subseteq \Theta$~\cite{srinivas2009gaussian}, $\gamma(J)$ can be approximated within a constant factor using a greedy algorithm~\cite{nemhauser1978analysis} prior to running \ours. Although the obtained value may be slightly larger than the exact $\gamma(J)$, it does not affect the correctness of Theorems~\ref{cla:bounds}-\ref{thm:main-T-ours}, as the confidence bounds become looser but still satisfy~\eqref{eq:lcb-ucb-property}.
For space, standard GP computations imply that calculating the per-query GP surrogates costs
$O(\sum_{q\in\cQ}J_{q,\tau}^2)=O(Q\cdot \tau^2)$, plus $O(\tau)$ to store the observation history.
Note that the empirical running time can be significantly lower than the above worst-case complexity, as $J_{\max,t}$ is usually smaller than $t$.
\section{Additional Related Work}
In this section, we review several lines of work that may seem related to ours, and explain how they differ from the setting in \ourproblem.

\stitle{Selecting a single model}
A first line of work selects a single model for a single task.
For instance, SpareLLM~\cite{jo2025sparellm} evaluates all candidate LLMs but may stop early based on a quality confidence interval.
BARGAIN~\cite{10.1145/3769776} considers metrics such as accuracy and recall, and proposes finer-grained estimation methods. These methods differ fundamentally from our setting, as they select a single model for an isolated task. In contrast, we select a specific LLM for each module within a compound AI system that handles multi-step tasks.

\stitle{Using multiple models at runtime}
Another line of work utilizes multiple LLMs via online routing or cascades, where queries are adaptively forwarded to different models during execution. FrugalGPT~\cite{chen2024frugalgpt} invokes models sequentially from cheapest to most expensive, stopping once a learned confidence criterion is met. LLM-Cascade~\cite{yue2024large} combines lightweight models with a decision mechanism and a stronger fallback model. ThriftLLM~\cite{10.14778/3749646.3749702} selects a subset of LLMs to maximize quality under a cost constraint by aggregating their outputs at runtime.
Similarly, the work by \citet{atalar2025neuralbanditbasedoptimal} studies model selection along a pipeline with partial feedback using bandit-style algorithms. For a comprehensive survey of using multiple LLMs during execution, we refer readers to~\citet{chen2025harnessingmultiplelargelanguage}.
However, these approaches target specific system architectures to enable query-dependent or dynamic routing. In contrast, our work focuses on configuring arbitrary compound AI systems. Furthermore, their objectives differ from ours: we minimize cost subject to a strict quality constraint.

\begin{table*}[t]
\centering
\caption{Statistics of tasks, systems, datasets, and search budgets.}
\label{tab:task-settings}
\small
\vspace{-2mm}
\renewcommand{\arraystretch}{1.1}
\begin{tabular}{lrrrrrrr}
\toprule
\textbf{Scenario} & \textbf{AI system} & \textbf{$N$} & \textbf{$|\Theta|$} & \textbf{Query dataset ($\cQ$)} & \textbf{$Q$} & \textbf{$\budget_{\rm max}$ in USD} & \textbf{Test-time dataset (for RQ2)} \\
\midrule
Text-to-SQL          & DIN-SQL~\cite{pourreza2023din}     &  4 & 279,841 & BIRD-mini-dev~\cite{10.5555/3666122.3667957} & 500 & 30.0 & BIRD-dev (1534 queries)~\cite{10.5555/3666122.3667957} \\
Data transformation   & UniDM-DT~\cite{qian2024unidm}    & 5 & 6,436,343 & Bing-QueryLogs~\cite{10.14778/3231751.3231766}  & 102 & 5.0 & StackOverflow (710 queries)~\cite{10.14778/3231751.3231766} \\
Data imputation      & UniDM-DI~\cite{qian2024unidm}    & 3 & 12,167 &  Restaurant-dev~\cite{9458712} & 156 & 2.0 & Restaurant-test (86 queries)~\cite{9458712} \\
\bottomrule
\end{tabular}
\end{table*}

\input{figure/main-dev}

\section{Experiments}\label{sec:exp}

In this section, we evaluate \ours on various tasks with different compound AI systems, addressing the following research questions:

\begin{itemize}[topsep=2pt,itemsep=1pt,parsep=0pt,partopsep=0pt,leftmargin=11pt]
\item \textbf{RQ1}: How does \ours improve correctness, effectiveness, and cost-efficiency compared to baseline approaches?
\item \textbf{RQ2}: Does the solution returned by \ours generalize in test-time deployments?
{\item \textbf{RQ3}: How sensitive is the performance of \ours to different problem and algorithmic settings?}
\end{itemize}
{We also conduct ablation studies on key components of \ours and additional scalability evaluations. Results are presented in 
\trref{\ref{appendix:ablation}}.}

\subsection{Experimental Setup}

\stitle{Problem settings}
We run experiments on three data management tasks: Text-to-SQL, data transformation, and data imputation. We use the compound AI systems in DIN-SQL~\cite{pourreza2023din}, UniDM-DT, and UniDM-DI~\cite{qian2024unidm} for these tasks, respectively.
For each task, we adopt the dataset $\cQ$ used in the system’s original paper, and set the maximum budget $\budget_{\max}$ to approximately match the monetary cost incurred when random search runs for 20 iterations, which is sufficient for most methods to converge to a stable solution. Details are listed in Table~\ref{tab:task-settings}.
In addition, we construct the candidate model set $\cM$ with 23 models covering a wide range of prices and capabilities, resulting in search spaces containing tens of thousands to millions of configurations. {\it The full LLM list is provided in \trref{\ref{appendix:experiments}}.} We set $\btheta_0$ to the configuration in which all modules use \texttt{GPT-5.2}, the most expensive model in $\cM$, and set $\epsilon=0.01$. We compute each $y_{g,t}$ using the standard evaluation metric of the task, i.e., execution accuracy for Text-to-SQL and accuracy for data transformation and data imputation. We compute each $y_{c,t}$ as the actual incurred cost in USD.
For any configuration $\btheta$ reported in our experiments, we estimate $c(\btheta)$ and $s(\btheta)$ offline by evaluating $\btheta$ on $\cQ$ multiple times and taking the average; these evaluations are not counted toward the budget of any method.

\stitle{Competitors and configurations}
We evaluate \ours against seven competitors, divided into two categories as follows.

\begin{itemize}[topsep=2pt,itemsep=1pt,parsep=0pt,partopsep=0pt,leftmargin=11pt]
\item {\it Generic optimizations}: random search (\rand), \BO~\cite{wang2025convergence}, \config~\cite{xu2023constrained}, and \safeopt~\cite{pmlr-v37-sui15}. 
For these methods, we follow the standard implementations from prior work~\cite{xu2023constrained}.

\item {\it Compound-AI-specific approaches}: \llambo~\cite{liu2024large}, \llmselect~\cite{chen2025optimizingmodelselectioncompound}, and \abacus~\cite{russo2025abacuscostbasedoptimizersemantic}. 
We make the necessary minimal adaptations to align with the setting in Problem~\ref{problem-main}; see \trref{\ref{appendix:experiments}} for details.
\end{itemize}
For \ours and generic approaches that require an SPD kernel $k$, we follow prior work~\cite{snoek2012practical} and use the Mat\'ern $5/2$ kernel:
\vspace{-3mm}
\begin{equation*}
\textstyle
k(\btheta,\btheta')=\left(1+\sqrt5\,d(\btheta,\btheta')+\tfrac53 d(\btheta,\btheta')^2\right)e^{-\sqrt5 d(\btheta,\btheta')},
\vspace{-3mm}
\end{equation*}
where $d(\btheta,\btheta')=\sqrt{\sum_{i=1}^{N}\bm 1\{\theta_i\neq \theta'_i\}}$ counts the number of modules on which $\btheta$ and $\btheta'$ choose different LLMs. To ensure a fair comparison, we use the same kernel across these methods, although kernel choices may affect empirical performance. For \ours, we set $R_c=R_g=10^{-3}$, $\delta=10^{-4}$, $\alpha=1/3$, and $\theta_{\rm base}$ to \texttt{Gemini-2.5-flash-lite}. We treat $B_c$ and $B_g$ as hyperparameters and tune them before the main loop (Line~\ref{alg:scope-main:init1} of Algorithm~\ref{alg:scope-main}), so that Line~\ref{alg:scope-main:select} of Algorithm~\ref{alg:scope-main} has at least one eligible configuration when executing for the first time.

\stitle{Evaluation metrics}
For each task and search budget $\budget \in [0,\budget_{\max}]$, we mainly report metrics for each algorithm's returned configuration, denoted as $\btheta_{\rm out,\budget}$.
{To evaluate effectiveness and cost-efficiency}, we report the {\it best feasible cost} achieved by each algorithm up to budget $\budget$, defined as
\[
c_{\rm bf}(\budget)=\min_{\budget'\leq \budget:~s(\btheta_{\rm out,\budget'})\ge s_0} c(\btheta_{\rm out,\budget'}).
\]
{Here, effectiveness is represented by how small the best feasible cost is at the maximum budget $\budget_{\max}$. Cost-efficiency is measured by how quickly the best feasible cost decreases as the budget increases.}
For \ours and its variants, $c_{\rm bf}(\budget)=\mathrm{SR}(\tau) + c(\btheta^\star)$ at the returned stopping time $\tau$ corresponding to budget $\budget$. For other methods, it similarly represents the best cost found up to budget $\budget$, where infeasible configurations are ruled out.
To assess {correctness, i.e., whether $s({\btheta}_{\rm out}) \geq s_0$}, we also report the average {\it constraint violation} metric~\cite{xu2023constrained} up to budget $\budget$, defined as 
\[
\textstyle
V(\budget)=\frac{1}{\budget}\int_{0}^{\budget}\frac{\max\{s_0-s(\btheta_{{\rm out}, u}),0\}}{s_0}du,
\]
with $V(0)=0$. $V(\budget)$ measures the extent to which the returned configuration violates the quality constraint, {and $V(\budget)=0$ means that correctness holds at all budgets from $0$ to $\budget$.} Since there are finitely many different $\btheta_{{\rm out}, u}$, we compute the integral exactly by summing over the corresponding budget intervals.
We repeat each experiment three times, each with a fixed random seed, and report the median, minimum, and maximum values. All experiments were conducted on a Linux server equipped with an Intel Xeon(R) Gold 6240 CPU @ 2.60GHz and 377GB of RAM, using multi-threading to maximize throughput when invoking LLMs. The implementations of all algorithms are available at {\color{blue}\url{https://github.com/waetr/SCOPE-LLM-optimizer/}}.


\newcommand{\NumW}{2.2em}
\newcommand{\ParenW}{3.6em} 
\newcommand{\GapW}{0.2em}
\newcommand{\costcell}[2]{%
  \makebox[\NumW][l]{#1}\hspace{\GapW}%
  \makebox[\ParenW][l]{\textcolor{Green}{(#2\%)}}%
}
\newcommand{\qposcell}[2]{%
  \makebox[\NumW][l]{#1}\hspace{\GapW}%
  \makebox[\ParenW][l]{\textcolor{Green}{(+#2\%)}}%
}
\newcommand{\qnegcell}[2]{%
  \makebox[\NumW][l]{#1}\hspace{\GapW}%
  \makebox[\ParenW][l]{\textcolor{Red}{(-#2\%)}}%
}

\begin{table*}[t]
\centering
\small
\caption{Avg. cost and quality on test-time datasets. The best and second-best values are \textbf{bold} and \underline{underlined}, respectively.}
\label{tab:test_performance}
\vspace{-2mm}
\begin{tabularx}{0.95\textwidth}{l*{6}{R}}
\toprule
\textbf{Method}
& \multicolumn{2}{c|}{\textbf{Text-to-SQL (BIRD-dev)}}
& \multicolumn{2}{c|}{\textbf{Data transformation (StackOverflow)}}
& \multicolumn{2}{c}{\textbf{Data imputation (Restaurant-test)}} \\
\cmidrule(lr){2-3}\cmidrule(lr){4-5}\cmidrule(lr){6-7}
& \multicolumn{1}{c}{\textbf{Avg.\ cost ($\times 10^{-2}$)}} & \multicolumn{1}{c|}{\textbf{Avg.\ quality}}
& \multicolumn{1}{c}{\textbf{Avg.\ cost ($\times 10^{-3}$)}} & \multicolumn{1}{c|}{\textbf{Avg.\ quality}}
& \multicolumn{1}{c}{\textbf{Avg.\ cost ($\times 10^{-3}$)}} & \multicolumn{1}{c}{\textbf{Avg.\ quality}} \\
\midrule
\textbf{Reference} ($\btheta_0$)
& \makebox[\NumW][l]{2.35}\hspace{\GapW}\makebox[\ParenW][l]{} & \multicolumn{1}{r|}{\makebox[\NumW][l]{0.34}\hspace{\GapW}\makebox[\ParenW][l]{}}
& \makebox[\NumW][l]{4.38}\hspace{\GapW}\makebox[\ParenW][l]{} & \multicolumn{1}{r|}{\makebox[\NumW][l]{0.37}\hspace{\GapW}\makebox[\ParenW][l]{}}
& \makebox[\NumW][l]{4.28}\hspace{\GapW}\makebox[\ParenW][l]{} & \makebox[\NumW][l]{0.74}\hspace{\GapW}\makebox[\ParenW][l]{} \\
\BO
& \costcell{1.05}{45} & \multicolumn{1}{r|}{\qposcell{0.35}{3}}
& \costcell{3.04}{69} & \multicolumn{1}{r|}{\qnegcell{0.34}{8}}
& \costcell{1.02}{24} & \qposcell{0.87}{18} \\
\config
& \costcell{1.41}{60} & \multicolumn{1}{r|}{\qposcell{0.34}{0}}
& \costcell{0.93}{21} & \multicolumn{1}{r|}{\qnegcell{0.24}{35}}
& \costcell{\underline{0.21}}{5} & \qposcell{\textbf{0.90}}{22} \\
\safeopt
& \costcell{2.35}{100} & \multicolumn{1}{r|}{\qposcell{0.34}{0}}
& \costcell{4.38}{100} & \multicolumn{1}{r|}{\qposcell{0.37}{0}}
& \costcell{1.16}{27} & \qposcell{\underline{0.88}}{19} \\
\llambo
& \costcell{2.35}{100} & \multicolumn{1}{r|}{\qposcell{0.34}{0}}
& \costcell{0.76}{17} & \multicolumn{1}{r|}{\qnegcell{0.20}{46}}
& \costcell{0.43}{10} & \qposcell{0.87}{18} \\
\llmselect
& \costcell{1.60}{68} & \multicolumn{1}{r|}{\qposcell{\underline{0.36}}{6}}
& \costcell{4.38}{100} & \multicolumn{1}{r|}{\qposcell{0.37}{0}}
& \costcell{4.28}{100} & \qposcell{0.74}{0} \\
\abacus
& \costcell{1.73}{74} & \multicolumn{1}{r|}{\qposcell{0.35}{3}}
& \costcell{4.38}{100} & \multicolumn{1}{r|}{\qposcell{0.37}{0}}
& \costcell{0.63}{15} & \qposcell{0.87}{18} \\
\rand
& \costcell{\underline{1.04}}{44} & \multicolumn{1}{r|}{\qposcell{0.35}{3}}
& \costcell{\underline{0.50}}{11} & \multicolumn{1}{r|}{\qposcell{\textbf{0.51}}{38}}
& \costcell{1.00}{23} & \qposcell{\underline{0.88}}{19} \\
{\bf \ours (Ours)}
& \costcell{\textbf{0.08}}{3} & \multicolumn{1}{r|}{\qposcell{\textbf{0.41}}{21}}
& \costcell{\textbf{0.22}}{5} & \multicolumn{1}{r|}{\qposcell{\underline{0.39}}{5}}
& \costcell{\textbf{0.10}}{2} & \qposcell{\underline{0.88}}{19} \\
\bottomrule
\end{tabularx}
\end{table*}

\subsection{Performance Evaluation (RQ1)}\label{sec:rq1}

In the first set of experiments, we evaluate \ours against competitors in the three tasks, varying the search budget $\budget \in [0,\budget_{\max}]$. Figure~\ref{fig:main-dev} reports the best feasible cost $c_{\rm bf}(\budget)$ and the constraint violation $V(\budget)$ of each algorithm. For better visualization, we report only the median of $V(\budget)$ for \llmselect across the three runs, as its maximum (up to $0.64$) would dominate most of the figure.

\stitle{Effectiveness and efficiency} 
As shown in the upper panels of Figure~\ref{fig:main-dev}, \ours consistently outperforms all competitors in terms of the final best feasible cost and how fast the cost converges. Specifically, at the maximum budget $\budget_{\max}$, $c_{\rm bf}(\budget_{\max})$ found by \ours is $83.7\%$, $53.5\%$, and $21.1\%$ lower than that of the best competitor in Text-to-SQL, data transformation, and data imputation, respectively. Furthermore, across budgets $\budget \in [0,\budget_{\max}]$, \ours achieves up to $95.5\%$, $63.1\%$, and $93.6\%$ lower best feasible cost $c_{\rm bf}(\budget)$ than the best competitor in the three tasks, respectively.
The lower best feasible cost at $\budget_{\max}$ suggests that the candidates $\btheta_{\rm cand}$ selected by \ours have strong cost and quality under the confidence-bound-based selection criterion.
During search, the rapid improvement in $c_{\rm bf}(\budget)$ is driven by \ours's exploitation of query-level observations: compared with competitors that rely on dataset-level evaluation and thus must spend a full pass over $\cQ$ for each configuration, \ours can rule out many infeasible or suboptimal candidates after evaluating only a few queries (e.g., in data imputation, we observe that 79\% of $\btheta_{\rm cand}$ considered by \ours are ruled out after evaluating the first 40 queries), saving substantial search budget. In addition, the \warmstart subroutine evaluates many good configurations at a small budget, unlike other competitors that often fail to obtain a good starting point during search.
Regarding other competitors, the performance among \BO, \config, \llambo, and \abacus varies substantially across different tasks, in contrast to the consistently strong performance of \ours.
\llmselect rarely finds a better solution than $\btheta_0$ across the three tasks, as it overemphasizes quality while overlooking cost.
\safeopt also performs conservatively, as it enforces that every evaluated configuration must be feasible, which limits exploration of potentially better ones. To summarize, \ours is significantly more cost-efficient and effective than all competitors.

\stitle{Correctness} As shown in the lower panels of Figure~\ref{fig:main-dev}, \ours exhibits \textit{zero constraint violations} $V(\budget)$ across all tasks and budgets, aligning with its $\delta$-correctness guarantee in theory. Regarding other methods, \safeopt also has zero violations and, in theory, ensures correctness under the same assumptions as \ours. \BO offers correctness only under a noiseless assumption; nonetheless, it also has zero violations in our experiments. The remaining methods exhibit oscillating or increasing $V(\budget)$, often due to the inductive biases of their estimation models, for which correctness guarantees are unknown.

\subsection{Deployment Performance (RQ2)}
In the second set of experiments, we evaluate, on a test-time dataset distinct from $\cQ$, the reported configuration at budget $\budget_{\rm max}$ in Figure~\ref{fig:main-dev} (i.e., the one achieving the best feasible cost on $\cQ$ within budget $\budget_{\rm max}$). The test-time datasets are listed in Table~\ref{tab:task-settings} and have been used as held-out test sets corresponding to $\cQ$~\cite{qian2024unidm}. Table~\ref{tab:test_performance} reports the resulting average cost and average quality on these datasets.
Overall, the average costs at test time are within the same order of magnitude as those on $\cQ$. In addition, some methods achieve cost reductions while maintaining average quality comparable to the reference configuration $\btheta_0$, demonstrating the presence of configurations with lower cost than the expensive single-model baseline without sacrificing potential quality.

\stitle{Performance of \ours} Notably, the configuration found by \ours on $\cQ$ generalizes well: \ours attains the lowest average cost among all methods while achieving higher average quality than the reference configuration $\btheta_0$. Concretely, in Text-to-SQL, data transformation, and data imputation, the average costs of \ours are only $3\%$, $5\%$, and $2\%$ of those of $\btheta_0$, respectively. Furthermore, its average quality improves upon $\btheta_0$ by $21\%$, $5\%$, and $19\%$. We attribute this surprising quality improvement to the way \ours ensures $\delta$-correctness under noisy evaluations. As feasibility must be certified with high probability, \ours is less likely to return configurations that sit right near the quality threshold on $\cQ$. Instead, it tends to select configurations with a meaningful quality margin (i.e., the part exceeding the quality threshold). Under dataset shift, this margin makes the selected configuration more likely to remain feasible and, in some cases, even to outperform $\btheta_0$ in average quality.

\stitle{Performance of other competitors}
In contrast to \ours, other competitors translate less reliably to test time. They achieve much smaller cost savings and may fail to match the average quality of $\btheta_0$. For example, \BO's average costs are $13.1\times$, $13.8\times$, and $10.2\times$ larger than those of \ours across the three tasks. In data transformation, its average quality is $8\%$ lower than that of $\btheta_0$. This reflects the brittleness of its correctness property under a noiseless assumption: by choosing solutions very close to the quality threshold, \BO can overfit to $\cQ$ rather than yield robust test-time performance. Similarly, \config and \llambo exhibit lower average quality than $\btheta_0$ in data transformation, as they lack mechanisms to generalize beyond $\cQ$. Meanwhile, while \safeopt maintains feasibility on these datasets, it realizes cost reductions much smaller than \ours, similar to its behavior during search. The remaining methods, \abacus and \llmselect, achieve average quality no worse than that of $\btheta_0$, but still incur average costs at least $6\times$ higher than \ours. In sum, \ours reduces more cost than other competitors, while delivering improved quality relative to the reference configuration on held-out queries.

\subsection{Sensitivity Analysis (RQ3)}

In the third set of experiments, we evaluate whether \ours's performance is sensitive to two crucial settings: the reference configuration $\btheta_{0}$ in Problem~\ref{problem-main} and the kernel $k(\cdot, \cdot)$ used in \ours. We conduct the experiments on data imputation, while keeping the other experimental settings the same as in the previous experiments. We compare the best feasible cost of \ours with that of the strongest competitors, \BO and \config, over the search budget $\budget$.
We first replace the default reference configuration with one in which all modules use \texttt{Claude~Haiku~4.5}. As shown in Figure~\ref{fig:sensitivity}(a), \ours remains effective under this change and attains the lowest best feasible cost once the budget exceeds $1.2$ USD. This shows that the improvement of \ours is not tied to a specific reference configuration.
We next replace the Mat\'ern $5/2$ kernel with the squared exponential (SE) kernel, defined as
$
k(\btheta,\btheta')=\exp(-d(\btheta,\btheta')^2/2).
$
As shown in Figure~\ref{fig:sensitivity}(b), \ours continues to achieve lower best feasible cost than the other kernel-based competitors without intensive kernel tuning. Specifically, at $\budget_{\rm max}$, the best feasible cost of \ours is at most $50\%$ of that of the second-best method. As a practical guideline, one can start with the default Mat\'ern $5/2$ kernel and use the SE kernel only if needed.

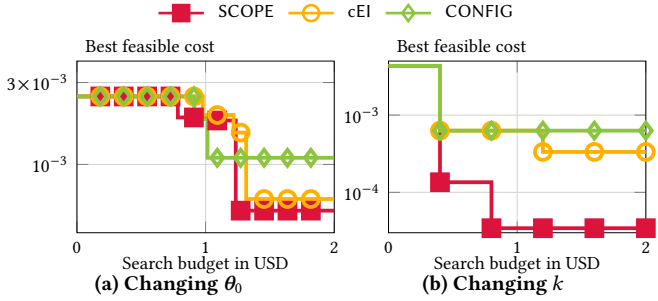
\begin{figure}[!t]
\centering
\begin{tikzpicture}
    \begin{customlegend}[legend columns=8,
        legend entries={\ours, \BO, \config, \safeopt}
        ,
        legend style={at={(0.45,1.05)},anchor=north,draw=none,font=\footnotesize,column sep=0.1cm}]
        
    \addlegendimage{line width=0.35mm,mark size=3pt,mark=square*,color=RedPPAlt}
    \addlegendimage{line width=0.35mm,mark size=3pt,mark=o,color=Orange}
    \addlegendimage{line width=0.35mm,mark size=3pt,mark=diamond,color=LightGreen}
    \end{customlegend}
\end{tikzpicture}
\\[-\lineskip]
\hspace{-4mm}
\subfloat[Changing $\btheta_0$]{
\begin{tikzpicture}[scale=1]
    \begin{axis}[
        axis line style={line width=0.3pt},
        height=\columnwidth/2.2,
        width=\columnwidth/1.7,
        ylabel={Best feasible cost},
        xmin=0, xmax=2.0,
        ymin=0.0004, ymax=0.004,
        xtick={0, 1,2,3,4,5},
        xticklabel style = {font=\footnotesize},
        xticklabels={0, 1,2,3,4,5},
        ymode=log,
        log basis y=10,
        grid=both,
major grid style={line width=0.3pt, draw=gray!35},
minor grid style={line width=0.2pt, draw=gray!20},
        ytick={3e-4,1e-3,3e-3},
        yticklabels={$3\!\times\!10^{-4}$,$10^{-3}$,$3\!\times\!10^{-3}$},
        every axis y label/.style={at={(current axis.north west)},anchor=south west, xshift=0mm, yshift=0mm},
        label style={font=\footnotesize},
        tick label style={font=\footnotesize},
after end axis/.code={
  \node[font=\footnotesize,anchor=north] at (rel axis cs:0.5,-0.09) {Search budget in USD};
},
        label style={font=\footnotesize},
        tick label style={font=\footnotesize},
    ]

\addplot[name path=SCOPEOptimizerUpper, draw=none] plot coordinates {
(0.000000, 0.002489)
(0.784312, 0.002489)
(0.784312, 0.001878)
(1.083170, 0.001878)
(1.083170, 0.001803)
(1.234462, 0.001803)
(1.234462, 0.000538)
(2.000000, 0.000538)
};
\addplot[name path=SCOPEOptimizerLower, draw=none] plot coordinates {
(0.000000, 0.002489)
(0.784312, 0.002489)
(0.784312, 0.001878)
(1.083170, 0.001878)
(1.083170, 0.001803)
(1.234462, 0.001803)
(1.234462, 0.000538)
(2.000000, 0.000538)
};
\addplot[fill=RedPPAlt, fill opacity=0.1, draw=none] fill between[of=SCOPEOptimizerUpper and SCOPEOptimizerLower];
\addplot[color=RedPPAlt, line width=0.5mm] plot coordinates {
(0.000000, 0.002489)
(0.784312, 0.002489)
(0.784312, 0.001878)
(1.083170, 0.001878)
(1.083170, 0.001803)
(1.234462, 0.001803)
(1.234462, 0.000538)
(2.000000, 0.000538)
};
\addplot[only marks, color=RedPPAlt, mark size=3.0pt, mark=square*, line width=0.5mm] coordinates {
(0.181818, 0.002489)
(0.363636, 0.002489)
(0.545455, 0.002489)
(0.727273, 0.002489)
(0.909091, 0.001878)
(1.090909, 0.001803)
(1.272727, 0.000538)
(1.454545, 0.000538)
(1.636364, 0.000538)
(1.818182, 0.000538)
};

\addplot[name path=BayesianOptimizerUpper, draw=none] plot coordinates {
(0.000000, 0.002489)
(0.978097, 0.002489)
(0.978097, 0.001941)
(1.217316, 0.001941)
(1.217316, 0.001533)
(1.315406, 0.001533)
(1.315406, 0.000629)
(2.000000, 0.000629)
};
\addplot[name path=BayesianOptimizerLower, draw=none] plot coordinates {
(0.000000, 0.002489)
(0.978097, 0.002489)
(0.978097, 0.001941)
(1.217316, 0.001941)
(1.217316, 0.001533)
(1.315406, 0.001533)
(1.315406, 0.000629)
(2.000000, 0.000629)
};
\addplot[fill=Orange, fill opacity=0.1, draw=none] fill between[of=BayesianOptimizerUpper and BayesianOptimizerLower];
\addplot[color=Orange, line width=0.5mm] plot coordinates {
(0.000000, 0.002489)
(0.978097, 0.002489)
(0.978097, 0.001941)
(1.217316, 0.001941)
(1.217316, 0.001533)
(1.315406, 0.001533)
(1.315406, 0.000629)
(2.000000, 0.000629)
};
\addplot[only marks, color=Orange, mark size=3.0pt, mark=o, line width=0.5mm] coordinates {
(0.181818, 0.002489)
(0.363636, 0.002489)
(0.545455, 0.002489)
(0.727273, 0.002489)
(0.909091, 0.002489)
(1.090909, 0.001941)
(1.272727, 0.001533)
(1.454545, 0.000629)
(1.636364, 0.000629)
(1.818182, 0.000629)
};

\addplot[name path=ConfigOptimizerUpper, draw=none] plot coordinates {
(0.000000, 0.002489)
(0.911759, 0.002489)
(0.911759, 0.001836)
(1.015116, 0.001836)
(1.015116, 0.001093)
(1.185574, 0.001093)
(1.185574, 0.001093)
(2.000000, 0.001093)
};
\addplot[name path=ConfigOptimizerLower, draw=none] plot coordinates {
(0.000000, 0.002489)
(0.911759, 0.002489)
(0.911759, 0.001836)
(1.015116, 0.001836)
(1.015116, 0.001093)
(1.185574, 0.001093)
(1.185574, 0.001093)
(2.000000, 0.001093)
};
\addplot[fill=LightGreen, fill opacity=0.1, draw=none] fill between[of=ConfigOptimizerUpper and ConfigOptimizerLower];
\addplot[color=LightGreen, line width=0.5mm] plot coordinates {
(0.000000, 0.002489)
(0.911759, 0.002489)
(0.911759, 0.001836)
(1.015116, 0.001836)
(1.015116, 0.001093)
(1.185574, 0.001093)
(1.185574, 0.001093)
(2.000000, 0.001093)
};
\addplot[only marks, color=LightGreen, mark size=3.0pt, mark=diamond, line width=0.5mm] coordinates {
(0.181818, 0.002489)
(0.363636, 0.002489)
(0.545455, 0.002489)
(0.727273, 0.002489)
(0.909091, 0.002489)
(1.090909, 0.001093)
(1.272727, 0.001093)
(1.454545, 0.001093)
(1.636364, 0.001093)
(1.818182, 0.001093)
};

    \end{axis}
\end{tikzpicture}%
}%
\hspace{-3mm}
\subfloat[Changing $k$]{
\begin{tikzpicture}[scale=1]
    \begin{axis}[
        axis line style={line width=0.3pt},
        height=\columnwidth/2.2,
        width=\columnwidth/1.7,
        ylabel={Best feasible cost},
        xmin=0, xmax=2.0,
        ymin=0.00003, ymax=0.005,
        xtick={0, 1,2,3,4,5},
        xticklabel style = {font=\footnotesize},
        xticklabels={0, 1,2,3,4,5},
        ymode=log,
        log basis y=10,
        grid=both,
major grid style={line width=0.3pt, draw=gray!35},
minor grid style={line width=0.2pt, draw=gray!20},
        ytick={1e-4,1e-3},
        yticklabels={$10^{-4}$,$10^{-3}$},
        every axis y label/.style={at={(current axis.north west)},anchor=south west, xshift=0mm, yshift=0mm},
        label style={font=\footnotesize},
        tick label style={font=\footnotesize},
after end axis/.code={
  \node[font=\footnotesize,anchor=north] at (rel axis cs:0.5,-0.09) {Search budget in USD};
},
        label style={font=\footnotesize},
        tick label style={font=\footnotesize},
    ]

\addplot[name path=SCOPEOptimizerUpper, draw=none] plot coordinates {
(0.000000, 0.004335894231)
(0.400000, 0.004335894231)
(0.400000, 0.000135196282)
(0.800000, 0.000135196282)
(0.800000, 0.000034351603)
(1.200000, 0.000034351603)
(1.600000, 0.000034351603)
(2.000000, 0.000034351603)
};
\addplot[name path=SCOPEOptimizerLower, draw=none] plot coordinates {
(0.000000, 0.004335894231)
(0.400000, 0.004335894231)
(0.400000, 0.000135196282)
(0.800000, 0.000135196282)
(0.800000, 0.000034351603)
(1.200000, 0.000034351603)
(1.600000, 0.000034351603)
(2.000000, 0.000034351603)
};
\addplot[fill=RedPPAlt, fill opacity=0.1, draw=none] fill between[of=SCOPEOptimizerUpper and SCOPEOptimizerLower];
\addplot[color=RedPPAlt, line width=0.5mm] plot coordinates {
(0.000000, 0.004335894231)
(0.400000, 0.004335894231)
(0.400000, 0.000135196282)
(0.800000, 0.000135196282)
(0.800000, 0.000034351603)
(1.200000, 0.000034351603)
(1.600000, 0.000034351603)
(2.000000, 0.000034351603)
};
\addplot[only marks, color=RedPPAlt, mark size=3.0pt, mark=square*, line width=0.5mm] coordinates {
(0.400000, 0.000135196282)
(0.800000, 0.000034351603)
(1.200000, 0.000034351603)
(1.600000, 0.000034351603)
(2.000000, 0.000034351603)
};

\addplot[name path=BayesianOptimizerUpper, draw=none] plot coordinates {
(0.000000, 0.004335894231)
(0.400000, 0.004335894231)
(0.400000, 0.000628868910)
(0.800000, 0.000628868910)
(1.200000, 0.000628868910)
(1.200000, 0.000333389615)
(1.600000, 0.000333389615)
(2.000000, 0.000333389615)
};
\addplot[name path=BayesianOptimizerLower, draw=none] plot coordinates {
(0.000000, 0.004335894231)
(0.400000, 0.004335894231)
(0.400000, 0.000628868910)
(0.800000, 0.000628868910)
(1.200000, 0.000628868910)
(1.200000, 0.000333389615)
(1.600000, 0.000333389615)
(2.000000, 0.000333389615)
};
\addplot[fill=Orange, fill opacity=0.1, draw=none] fill between[of=BayesianOptimizerUpper and BayesianOptimizerLower];
\addplot[color=Orange, line width=0.5mm] plot coordinates {
(0.000000, 0.004335894231)
(0.400000, 0.004335894231)
(0.400000, 0.000628868910)
(0.800000, 0.000628868910)
(1.200000, 0.000628868910)
(1.200000, 0.000333389615)
(1.600000, 0.000333389615)
(2.000000, 0.000333389615)
};
\addplot[only marks, color=Orange, mark size=3.0pt, mark=o, line width=0.5mm] coordinates {
(0.400000, 0.000628868910)
(0.800000, 0.000628868910)
(1.200000, 0.000333389615)
(1.600000, 0.000333389615)
(2.000000, 0.000333389615)
};

\addplot[name path=ConfigOptimizerUpper, draw=none] plot coordinates {
(0.000000, 0.004335894231)
(0.400000, 0.004335894231)
(0.400000, 0.000628868910)
(0.800000, 0.000628868910)
(1.200000, 0.000628868910)
(1.600000, 0.000628868910)
(2.000000, 0.000628868910)
};
\addplot[name path=ConfigOptimizerLower, draw=none] plot coordinates {
(0.000000, 0.004335894231)
(0.400000, 0.004335894231)
(0.400000, 0.000628868910)
(0.800000, 0.000628868910)
(1.200000, 0.000628868910)
(1.600000, 0.000628868910)
(2.000000, 0.000628868910)
};
\addplot[fill=LightGreen, fill opacity=0.1, draw=none] fill between[of=ConfigOptimizerUpper and ConfigOptimizerLower];
\addplot[color=LightGreen, line width=0.5mm] plot coordinates {
(0.000000, 0.004335894231)
(0.400000, 0.004335894231)
(0.400000, 0.000628868910)
(0.800000, 0.000628868910)
(1.200000, 0.000628868910)
(1.600000, 0.000628868910)
(2.000000, 0.000628868910)
};
\addplot[only marks, color=LightGreen, mark size=3.0pt, mark=diamond, line width=0.5mm] coordinates {
(0.400000, 0.000628868910)
(0.800000, 0.000628868910)
(1.200000, 0.000628868910)
(1.600000, 0.000628868910)
(2.000000, 0.000628868910)
};

    \end{axis}
\end{tikzpicture}%
}%
\vspace{-3mm}
\caption{Best feasible cost when changing reference configurations and kernels.} \label{fig:sensitivity}
\vspace{-2mm}
\end{figure}
\section{Conclusion}
We study \ourproblem. To solve this problem, we propose \ours, a search algorithm that exploits query-level partial evaluations and provides high-probability feasibility and simple-regret guarantees. In experiments, \ours reduces search cost relative to state-of-the-art baselines while better satisfying the quality constraint across multiple tasks. For future work, we plan to extend \ours to richer representations of cost--quality trade-offs, including multi-objective optimization and multiple constraints.

\begin{acks}
This work was supported by the Ministry of Education, Singapore, under Tier-2 Grant MOE-000761-01.
\end{acks}

\bibliographystyle{ACM-Reference-Format}
\bibliography{refs}

\appendix
\section{Experiment Details}\label{appendix:experiments}

\begin{figure*}[!t]
\captionsetup[subfloat]{labelformat=empty}
\centering
\begin{tikzpicture}
    \begin{customlegend}[legend columns=3,
        legend entries={\ours, \ggours, \gours}
        ,
        legend style={at={(0.45,1.05)},anchor=north,draw=none,font=\footnotesize,column sep=0.1cm}]
        
    \addlegendimage{line width=0.35mm,mark size=3pt,mark=square*,color=RedPPAlt}
    \addlegendimage{line width=0.35mm,mark size=3pt,mark=square,color=LightGreen}
    \addlegendimage{line width=0.35mm,mark size=3pt,mark=asterisk,color=Gray}
    \end{customlegend}
\end{tikzpicture}
\\[-\lineskip]
\hspace{-2mm}
\subfloat[]{
\begin{tikzpicture}[scale=1]
    \begin{axis}[
        axis line style={line width=0.3pt},
        height=\columnwidth/2.0,
        width=\columnwidth/1.4,
        ylabel={Best feasible cost},
after end axis/.code={
  \node[font=\footnotesize,anchor=north] at (rel axis cs:0.5,-0.09) {Search budget in USD};
  \node[font=\footnotesize,anchor=north] at (rel axis cs:0.5,-0.23) {{\bf (a) Text-to-SQL}};
},
        xmin=0, xmax=30.0,
        ymin=0.0007, ymax=0.03,
        xtick={0, 10, 20, 30},
        xticklabel style = {font=\footnotesize},
        xticklabels={0, 10, 20, 30},
        ytick={0.001, 0.003, 0.01, 0.03}, 
        yticklabels={$10^{-3}$, $3\times10^{-3}$, $10^{-2}$, $3\times 10^{-2}$}, 
        ymode=log,
        log basis y=10,
        grid=both,
major grid style={line width=0.3pt, draw=Gray!35},
minor grid style={line width=0.2pt, draw=Gray!20},
        every axis y label/.style={at={(current axis.north west)},anchor=south west, xshift=0mm, yshift=0mm},
        label style={font=\footnotesize},
        tick label style={font=\footnotesize},
        label style={font=\footnotesize},
        tick label style={font=\footnotesize},
    ]

\addplot[name path=SCOPEOptimizerUpper, draw=none] plot coordinates {
(0.000000, 0.024519)
(27.631037, 0.024519)
(27.631037, 0.013915)
(30.000000, 0.013915)
};
\addplot[name path=SCOPEOptimizerLower, draw=none] plot coordinates {
(0.000000, 0.024519)
(17.236446, 0.024519)
(17.236446, 0.015335)
(22.106070, 0.015335)
(22.106070, 0.001519)
(22.790217, 0.001519)
(22.790217, 0.001297)
(24.427746, 0.001297)
(24.427746, 0.001009)
(30.000000, 0.001009)
};
\addplot[fill=Gray, fill opacity=0.1, draw=none] fill between[of=SCOPEOptimizerUpper and SCOPEOptimizerLower];
\addplot[color=Gray, line width=0.5mm] plot coordinates {
(0.000000, 0.024519)
(22.106070, 0.024519)
(22.106070, 0.015335)
(24.193782, 0.015335)
(24.193782, 0.013915)
(27.631037, 0.013915)
(27.631037, 0.010603)
(29.628259, 0.010603)
(29.628259, 0.003994)
(30.000000, 0.003994)
};
\addplot[only marks, color=Gray, mark size=3.0pt, mark=asterisk, line width=0.5mm] coordinates {
(2.727273, 0.024519)
(5.454545, 0.024519)
(8.181818, 0.024519)
(10.909091, 0.024519)
(13.636364, 0.024519)
(16.363636, 0.024519)
(19.090909, 0.024519)
(21.818182, 0.024519)
(24.545455, 0.013915)
(27.272727, 0.013915)
};

\addplot[name path=SCOPEOptimizerMMUpper, draw=none] plot coordinates {
(0.000000, 0.024519)
(11.928993, 0.024519)
(11.928993, 0.005083)
(12.667410, 0.005083)
(12.667410, 0.004505)
(13.503384, 0.004505)
(13.503384, 0.002065)
(14.783646, 0.002065)
(14.783646, 0.001976)
(21.194281, 0.001976)
(21.194281, 0.001963)
(22.092770, 0.001963)
(22.092770, 0.001959)
(30.000000, 0.001959)
};
\addplot[name path=SCOPEOptimizerMMLower, draw=none] plot coordinates {
(0.000000, 0.024519)
(7.698171, 0.024519)
(7.698171, 0.002139)
(8.763607, 0.002139)
(8.763607, 0.001570)
(9.658568, 0.001570)
(9.658568, 0.001483)
(13.740767, 0.001483)
(13.740767, 0.001435)
(24.133253, 0.001435)
(24.133253, 0.001092)
(30.000000, 0.001092)
};
\addplot[fill=LightGreen, fill opacity=0.1, draw=none] fill between[of=SCOPEOptimizerMMUpper and SCOPEOptimizerMMLower];
\addplot[color=LightGreen, line width=0.5mm] plot coordinates {
(0.000000, 0.024519)
(10.962798, 0.024519)
(10.962798, 0.004505)
(12.667410, 0.004505)
(12.667410, 0.001976)
(14.783646, 0.001976)
(14.783646, 0.001895)
(18.797852, 0.001895)
(18.797852, 0.001519)
(30.000000, 0.001519)
};
\addplot[only marks, color=LightGreen, mark size=3.0pt, mark=square, line width=0.5mm] coordinates {
(2.727273, 0.024519)
(5.454545, 0.024519)
(8.181818, 0.024519)
(10.909091, 0.024519)
(13.636364, 0.001976)
(16.363636, 0.001895)
(19.090909, 0.001519)
(21.818182, 0.001519)
(24.545455, 0.001519)
(27.272727, 0.001519)
};

\addplot[name path=SCOPEOptimizerPPUpper, draw=none] plot coordinates {
(0.000000, 0.024519)
(5.930226, 0.024519)
(5.930226, 0.002609)
(6.010100, 0.002609)
(6.010100, 0.002339)
(8.520944, 0.002339)
(8.520944, 0.002067)
(9.552145, 0.002067)
(9.552145, 0.001789)
(10.641727, 0.001789)
(10.641727, 0.001765)
(11.563596, 0.001765)
(11.563596, 0.001716)
(24.726841, 0.001716)
(24.726841, 0.001106)
(30.000000, 0.001106)
};
\addplot[name path=SCOPEOptimizerPPLower, draw=none] plot coordinates {
(0.000000, 0.024519)
(4.657106, 0.024519)
(4.657106, 0.002609)
(4.992228, 0.002609)
(4.992228, 0.001946)
(6.274129, 0.001946)
(6.274129, 0.001604)
(7.207738, 0.001604)
(7.207738, 0.001289)
(8.125889, 0.001289)
(8.125889, 0.001106)
(9.942451, 0.001106)
(9.942451, 0.001009)
(19.784206, 0.001009)
(19.784206, 0.000829)
(20.168931, 0.000829)
(20.168931, 0.000769)
(30.000000, 0.000769)
};
\addplot[fill=RedPPAlt, fill opacity=0.1, draw=none] fill between[of=SCOPEOptimizerPPUpper and SCOPEOptimizerPPLower];
\addplot[color=RedPPAlt, line width=0.5mm] plot coordinates {
(0.000000, 0.024519)
(4.992228, 0.024519)
(4.992228, 0.002609)
(5.930226, 0.002609)
(5.930226, 0.001946)
(8.044100, 0.001946)
(8.044100, 0.001562)
(9.048987, 0.001562)
(9.048987, 0.001289)
(9.942451, 0.001289)
(9.942451, 0.001106)
(10.884634, 0.001106)
(10.884634, 0.001098)
(19.784206, 0.001098)
(19.784206, 0.001009)
(20.541891, 0.001009)
(20.541891, 0.000769)
(30.000000, 0.000769)
};
\addplot[only marks, color=RedPPAlt, mark size=3.0pt, mark=square*, line width=0.5mm] coordinates {
(2.727273, 0.024519)
(5.454545, 0.002609)
(8.181818, 0.001562)
(10.909091, 0.001098)
(13.636364, 0.001098)
(16.363636, 0.001098)
(19.090909, 0.001098)
(21.818182, 0.000769)
(24.545455, 0.000769)
(27.272727, 0.000769)
};

    \end{axis}
\end{tikzpicture}%
}%
\hspace{-2mm}
\subfloat[]{
\begin{tikzpicture}[scale=1]
    \begin{axis}[
        axis line style={line width=0.3pt},
        height=\columnwidth/2.0,
        width=\columnwidth/1.4,
        ylabel={Best feasible cost},
        xmin=0, xmax=5.0,
        ymin=0.0003, ymax=0.005,
        xtick={0, 1,2,3,4,5},
        xticklabel style = {font=\footnotesize},
        xticklabels={0, 1,2,3,4,5},
        ymode=log,
        log basis y=10,
        grid=both,
major grid style={line width=0.3pt, draw=Gray!35},
minor grid style={line width=0.2pt, draw=Gray!20},
        ytick={3e-4,1e-3,3e-3},
        yticklabels={$3\!\times\!10^{-4}$,$10^{-3}$,$3\!\times\!10^{-3}$},
        every axis y label/.style={at={(current axis.north west)},anchor=south west, xshift=0mm, yshift=0mm},
        label style={font=\footnotesize},
        tick label style={font=\footnotesize},
after end axis/.code={
  \node[font=\footnotesize,anchor=north] at (rel axis cs:0.5,-0.09) {Search budget in USD};
  \node[font=\footnotesize,anchor=north] at (rel axis cs:0.5,-0.23) {{\bf (b) Data transformation}};
},
        label style={font=\footnotesize},
        tick label style={font=\footnotesize},
    ]

\addplot[name path=SCOPEOptimizerUpper, draw=none] plot coordinates {
(0.000000, 0.004516)
(1.815252, 0.004516)
(1.815252, 0.003396)
(2.517342, 0.003396)
(2.517342, 0.002722)
(3.673903, 0.002722)
(3.673903, 0.002523)
(4.407199, 0.002523)
(4.407199, 0.002166)
(4.617870, 0.002166)
(4.617870, 0.002065)
(4.873382, 0.002065)
(4.873382, 0.001156)
(5.000000, 0.001156)
};
\addplot[name path=SCOPEOptimizerLower, draw=none] plot coordinates {
(0.000000, 0.004516)
(0.704464, 0.004516)
(0.704464, 0.001373)
(1.168346, 0.001373)
(1.168346, 0.001156)
(3.671989, 0.001156)
(3.671989, 0.000864)
(5.000000, 0.000864)
};
\addplot[fill=Gray, fill opacity=0.1, draw=none] fill between[of=SCOPEOptimizerUpper and SCOPEOptimizerLower];
\addplot[color=Gray, line width=0.5mm] plot coordinates {
(0.000000, 0.004516)
(1.168346, 0.004516)
(1.168346, 0.001373)
(3.671989, 0.001373)
(3.671989, 0.001156)
(4.873382, 0.001156)
(4.873382, 0.001090)
(5.000000, 0.001090)
};
\addplot[only marks, color=Gray, mark size=3.0pt, mark=asterisk, line width=0.5mm] coordinates {
(0.454545, 0.004516)
(0.909091, 0.004516)
(1.363636, 0.001373)
(1.818182, 0.001373)
(2.272727, 0.001373)
(2.727273, 0.001373)
(3.181818, 0.001373)
(3.636364, 0.001373)
(4.090909, 0.001156)
(4.545455, 0.001156)
};

\addplot[name path=SCOPEOptimizerMMUpper, draw=none] plot coordinates {
(0.000000, 0.004516)
(2.053956, 0.004516)
(2.053956, 0.004024)
(2.601319, 0.004024)
(2.601319, 0.003615)
(2.687637, 0.003615)
(2.687637, 0.002638)
(2.896060, 0.002638)
(2.896060, 0.001582)
(4.556213, 0.001582)
(4.556213, 0.000951)
(4.823824, 0.000951)
(4.823824, 0.000933)
(4.916882, 0.000933)
(4.916882, 0.000866)
(5.000000, 0.000866)
};
\addplot[name path=SCOPEOptimizerMMLower, draw=none] plot coordinates {
(0.000000, 0.004516)
(0.911632, 0.004516)
(0.911632, 0.001074)
(1.051519, 0.001074)
(1.051519, 0.000857)
(1.198208, 0.000857)
(1.198208, 0.000571)
(1.264115, 0.000571)
(1.264115, 0.000497)
(2.235518, 0.000497)
(2.235518, 0.000373)
(3.544830, 0.000373)
(3.544830, 0.000262)
(3.785925, 0.000262)
(3.785925, 0.000241)
(5.000000, 0.000241)
};
\addplot[fill=LightGreen, fill opacity=0.1, draw=none] fill between[of=SCOPEOptimizerMMUpper and SCOPEOptimizerMMLower];
\addplot[color=LightGreen, line width=0.5mm] plot coordinates {
(0.000000, 0.004516)
(1.759096, 0.004516)
(1.759096, 0.002906)
(2.390558, 0.002906)
(2.390558, 0.002638)
(2.687637, 0.002638)
(2.687637, 0.000846)
(2.760356, 0.000846)
(2.760356, 0.000653)
(3.344344, 0.000653)
(3.344344, 0.000491)
(4.053135, 0.000491)
(4.053135, 0.000429)
(5.000000, 0.000429)
};
\addplot[only marks, color=LightGreen, mark size=3.0pt, mark=square, line width=0.5mm] coordinates {
(0.454545, 0.004516)
(0.909091, 0.004516)
(1.363636, 0.004516)
(1.818182, 0.002906)
(2.272727, 0.002906)
(2.727273, 0.000846)
(3.181818, 0.000653)
(3.636364, 0.000491)
(4.090909, 0.000429)
(4.545455, 0.000429)
};

\addplot[name path=SCOPEOptimizerPPUpper, draw=none] plot coordinates {
(0.000000, 0.004516)
(0.823436, 0.004516)
(0.823436, 0.001839)
(1.799169, 0.001839)
(1.799169, 0.001117)
(1.866028, 0.001117)
(1.866028, 0.000614)
(2.231526, 0.000614)
(2.231526, 0.000610)
(2.321885, 0.000610)
(2.321885, 0.000566)
(2.495342, 0.000566)
(2.495342, 0.000538)
(2.605036, 0.000538)
(2.605036, 0.000510)
(3.080814, 0.000510)
(3.080814, 0.000482)
(3.769174, 0.000482)
(3.769174, 0.000468)
(4.135718, 0.000468)
(4.135718, 0.000456)
(5.000000, 0.000456)
};
\addplot[name path=SCOPEOptimizerPPLower, draw=none] plot coordinates {
(0.000000, 0.004516)
(0.556960, 0.004516)
(0.556960, 0.001839)
(0.926772, 0.001839)
(0.926772, 0.000614)
(1.151860, 0.000614)
(1.151860, 0.000459)
(1.210739, 0.000459)
(1.210739, 0.000424)
(1.870783, 0.000424)
(1.870783, 0.000392)
(2.077924, 0.000392)
(2.077924, 0.000327)
(2.996593, 0.000327)
(2.996593, 0.000324)
(3.008732, 0.000324)
(3.008732, 0.000317)
(3.039793, 0.000317)
(3.039793, 0.000305)
(5.000000, 0.000305)
};
\addplot[fill=RedPPAlt, fill opacity=0.1, draw=none] fill between[of=SCOPEOptimizerPPUpper and SCOPEOptimizerPPLower];
\addplot[color=RedPPAlt, line width=0.5mm] plot coordinates {
(0.000000, 0.004516)
(0.594944, 0.004516)
(0.594944, 0.001839)
(1.151860, 0.001839)
(1.151860, 0.000614)
(1.866028, 0.000614)
(1.866028, 0.000484)
(2.077924, 0.000484)
(2.077924, 0.000392)
(2.236652, 0.000392)
(2.236652, 0.000363)
(2.587799, 0.000363)
(2.587799, 0.000356)
(2.765463, 0.000356)
(2.765463, 0.000355)
(2.883440, 0.000355)
(2.883440, 0.000329)
(2.996593, 0.000329)
(2.996593, 0.000327)
(3.008732, 0.000327)
(3.008732, 0.000324)
(4.581567, 0.000324)
(4.581567, 0.000310)
(5.000000, 0.000310)
};
\addplot[only marks, color=RedPPAlt, mark size=3.0pt, mark=square*, line width=0.5mm] coordinates {
(0.454545, 0.004516)
(0.909091, 0.001839)
(1.363636, 0.000614)
(1.818182, 0.000614)
(2.272727, 0.000363)
(2.727273, 0.000356)
(3.181818, 0.000324)
(3.636364, 0.000324)
(4.090909, 0.000324)
(4.545455, 0.000324)
};

    \end{axis}
\end{tikzpicture}%
}%
\hspace{-2mm}
\subfloat[]{
\begin{tikzpicture}[scale=1]
    \begin{axis}[
        axis line style={line width=0.3pt},
        height=\columnwidth/2.0,
        width=\columnwidth/1.4,
        ylabel={Best feasible cost},
        xmin=0, xmax=2.0,
        ymin=0.00005, ymax=0.005,
        xtick={0, 0.5, 1, 1.5, 2},
        xticklabel style = {font=\footnotesize},
        xticklabels={0, 0.5, 1, 1.5, 2},
        ytick={0.0001, 0.0003, 0.001, 0.003}, 
        yticklabels={$10^{-4}$, $3\times 10^{-4}$, $10^{-3}$, $3\times 10^{-3}$}, 
        ymode=log,
        log basis y=10,
        grid=both,
major grid style={line width=0.3pt, draw=Gray!35},
minor grid style={line width=0.2pt, draw=Gray!20},
        every axis y label/.style={at={(current axis.north west)},anchor=south west, xshift=0mm, yshift=0mm},
        label style={font=\footnotesize},
        tick label style={font=\footnotesize},
after end axis/.code={
  \node[font=\footnotesize,anchor=north] at (rel axis cs:0.5,-0.09) {Search budget in USD};
  \node[font=\footnotesize,anchor=north] at (rel axis cs:0.5,-0.23) {{\bf (c) Data imputation}};
},
        label style={font=\footnotesize},
        tick label style={font=\footnotesize},
    ]

\addplot[name path=SCOPEOptimizerUpper, draw=none] plot coordinates {
(0.000000, 0.004336)
(0.998018, 0.004336)
(0.998018, 0.000967)
(1.069968, 0.000967)
(1.069968, 0.000796)
(1.330479, 0.000796)
(1.330479, 0.000707)
(1.592457, 0.000707)
(1.592457, 0.000433)
(2.000000, 0.000433)
};
\addplot[name path=SCOPEOptimizerLower, draw=none] plot coordinates {
(0.000000, 0.004336)
(0.590422, 0.004336)
(0.590422, 0.001451)
(0.735358, 0.001451)
(0.735358, 0.000796)
(0.854087, 0.000796)
(0.854087, 0.000707)
(1.069968, 0.000707)
(1.069968, 0.000413)
(1.330479, 0.000413)
(1.330479, 0.000363)
(1.518018, 0.000363)
(1.518018, 0.000333)
(2.000000, 0.000333)
};
\addplot[fill=Gray, fill opacity=0.1, draw=none] fill between[of=SCOPEOptimizerUpper and SCOPEOptimizerLower];
\addplot[color=Gray, line width=0.5mm] plot coordinates {
(0.000000, 0.004336)
(0.741162, 0.004336)
(0.741162, 0.001451)
(0.854087, 0.001451)
(0.854087, 0.000796)
(1.069968, 0.000796)
(1.069968, 0.000707)
(1.330479, 0.000707)
(1.330479, 0.000413)
(1.518018, 0.000413)
(1.518018, 0.000363)
(2.000000, 0.000363)
};
\addplot[only marks, color=Gray, mark size=3.0pt, mark=asterisk, line width=0.5mm] coordinates {
(0.181818, 0.004336)
(0.363636, 0.004336)
(0.545455, 0.004336)
(0.727273, 0.004336)
(0.909091, 0.000796)
(1.090909, 0.000707)
(1.272727, 0.000707)
(1.454545, 0.000413)
(1.636364, 0.000363)
(1.818182, 0.000363)
};

\addplot[name path=SCOPEOptimizerMMUpper, draw=none] plot coordinates {
(0.000000, 0.004336)
(0.601365, 0.004336)
(0.601365, 0.001241)
(0.663494, 0.001241)
(0.663494, 0.001168)
(0.695780, 0.001168)
(0.695780, 0.000621)
(0.742812, 0.000621)
(0.742812, 0.000308)
(0.966672, 0.000308)
(0.966672, 0.000285)
(1.475874, 0.000285)
(1.475874, 0.000176)
(2.000000, 0.000176)
};
\addplot[name path=SCOPEOptimizerMMLower, draw=none] plot coordinates {
(0.000000, 0.004336)
(0.562416, 0.004336)
(0.562416, 0.001168)
(0.591483, 0.001168)
(0.591483, 0.000176)
(0.826407, 0.000176)
(0.826407, 0.000055)
(0.869898, 0.000055)
(0.869898, 0.000034)
(2.000000, 0.000034)
};
\addplot[name path=SCOPEOptimizerMMLower, draw=none] plot coordinates {
(0.000000, 0.004336)
(0.562416, 0.004336)
(0.562416, 0.001168)
(0.591483, 0.001168)
(0.591483, 0.000176)
(0.826407, 0.000176)
(0.826407, 0.000055)
(0.869898, 0.000055)
(0.869898, 0.000034)
(2.000000, 0.000034)
};
\addplot[fill=LightGreen, fill opacity=0.1, draw=none] fill between[of=SCOPEOptimizerMMUpper and SCOPEOptimizerMMLower];
\addplot[color=LightGreen, line width=0.5mm] plot coordinates {
(0.000000, 0.004336)
(0.569374, 0.004336)
(0.569374, 0.001239)
(0.591483, 0.001239)
(0.591483, 0.001168)
(0.663494, 0.001168)
(0.663494, 0.000287)
(0.826407, 0.000287)
(0.826407, 0.000176)
(0.837364, 0.000176)
(0.837364, 0.000062)
(0.869898, 0.000062)
(0.869898, 0.000055)
(2.000000, 0.000055)
};
\addplot[only marks, color=LightGreen, mark size=3.0pt, mark=square, line width=0.5mm] coordinates {
(0.181818, 0.004336)
(0.363636, 0.004336)
(0.545455, 0.004336)
(0.727273, 0.000287)
(0.909091, 0.000055)
(1.090909, 0.000055)
(1.272727, 0.000055)
(1.454545, 0.000055)
(1.636364, 0.000055)
(1.818182, 0.000055)
};

\addplot[name path=SCOPEOptimizerPPUpper, draw=none] plot coordinates {
(0.000000, 0.004336)
(0.364637, 0.004336)
(0.364637, 0.000824)
(0.469961, 0.000824)
(0.469961, 0.000620)
(0.579613, 0.000620)
(0.579613, 0.000536)
(0.741463, 0.000536)
(0.741463, 0.000203)
(1.119869, 0.000203)
(1.119869, 0.000183)
(1.166371, 0.000183)
(1.166371, 0.000135)
(1.174879, 0.000135)
(1.174879, 0.000112)
(2.000000, 0.000112)
};
\addplot[name path=SCOPEOptimizerPPLower, draw=none] plot coordinates {
(0.000000, 0.004336)
(0.247597, 0.004336)
(0.247597, 0.000238)
(0.323329, 0.000238)
(0.323329, 0.000112)
(1.174879, 0.000112)
(1.174879, 0.000055)
(1.202525, 0.000055)
(1.202525, 0.000040)
(1.207589, 0.000040)
(1.207589, 0.000024)
(2.000000, 0.000024)
};
\addplot[fill=RedPPAlt, fill opacity=0.1, draw=none] fill between[of=SCOPEOptimizerPPUpper and SCOPEOptimizerPPLower];
\addplot[color=RedPPAlt, line width=0.5mm] plot coordinates {
(0.000000, 0.004336)
(0.273729, 0.004336)
(0.273729, 0.000278)
(0.368482, 0.000278)
(0.368482, 0.000203)
(0.741463, 0.000203)
(0.741463, 0.000112)
(2.000000, 0.000112)
};
\addplot[only marks, color=RedPPAlt, mark size=3.0pt, mark=square*, line width=0.5mm] coordinates {
(0.181818, 0.004336)
(0.363636, 0.000278)
(0.545455, 0.000203)
(0.727273, 0.000203)
(0.909091, 0.000112)
(1.090909, 0.000112)
(1.272727, 0.000112)
(1.454545, 0.000112)
(1.636364, 0.000112)
(1.818182, 0.000112)
};
    
    \end{axis}
\end{tikzpicture}%
}%
\vspace{-5mm}
\caption{Best feasible cost of \ours and its variants.} \label{fig:main-abl}
\vspace{-1mm}
\end{figure*}
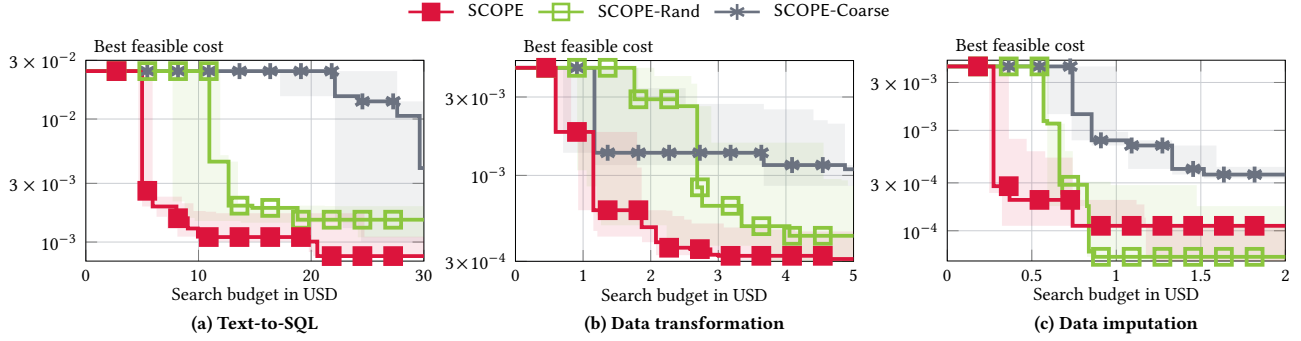

\stitle{Candidate LLMs}
The candidate LLMs used in the experiments are listed in Table~\ref{tab:llm-costs}. The pricing values are obtained from the official OpenAI, Google, Anthropic, and DeepInfra platforms as of the submission date. According to these platforms, the cost incurred by invoking an LLM equals (\# input tokens) $\times$ (input price) $+$ (\# output tokens) $\times$ (output price).

\begin{table}[h]
\centering
\caption{Candidate LLMs $\cM$ (price: USD per 1M tokens).}
\label{tab:llm-costs}
\small
\begin{tabular}{lrr}
\toprule
\textbf{Model} & \textbf{Input Price} & \textbf{Output Price} \\
\midrule
\textbf{GPT-5.2}                & \textbf{\$1.75}  & \textbf{\$14.00} \\
GPT-5-mini                & \$0.25  & \$2.00 \\
GPT-5-nano                & \$0.05  & \$0.40 \\
GPT-4.1               & \$2.00  & \$8.00  \\
GPT-4.1 Mini          & \$0.40  & \$1.60  \\
GPT-4.1 Nano          & \$0.10  & \$0.40  \\
Gemini 3 Flash      & \$0.50  & \$3.00  \\
Gemini 2.5 Flash      & \$0.30  & \$2.50  \\
Gemini 2.5 Flash-Lite & \$0.10  & \$0.40  \\
Gemini 2.0 Flash-Lite & \$0.08 & \$0.30  \\
Claude Haiku 4.5     & \$1.00  & \$5.00 \\
Claude Haiku 3.5      & \$0.80  & \$4.00  \\
Claude Haiku 3        & \$0.25  & \$1.25  \\
DeepSeek-V3.2              & \$0.26  & \$0.39 \\
DeepSeek-V3.1-Terminus     & \$0.21  & \$0.79 \\
Qwen3-235B-A22B        & \$0.07 & \$0.46 \\
Qwen3-Next-80B-A3B     & \$0.09  & \$1.10  \\
gemma-3-27b                & \$0.09  & \$0.16  \\
gemma-3-12b                & \$0.04  & \$0.13  \\
gemma-3-4b                & \$0.04  & \$0.08  \\
Mistral-Small-3.2      & \$0.08 & \$0.20  \\
Mistral-Small-3        & \$0.05  & \$0.08  \\
Mistral-Nemo           & \$0.02  & \$0.04  \\
\bottomrule
\end{tabular}
\end{table}

\stitle{Adaptation details}
\llambo uses natural language instructions to prompt an LLM to emulate the BO process, where each instruction includes metadata about the compound AI system, the dataset, the candidate LLMs, and the past observation history. We adapt the instruction template from the original paper~\cite{liu2024large} to extend it to the constrained setting. The underlying LLM for estimation is \texttt{GPT-5.2}. \llmselect~\cite{chen2025optimizingmodelselectioncompound} maximizes the quality $s$ by sequentially updating the LLM configuration starting from a random configuration, without considering cost. In addition, it optionally adopts an LLM diagnostician to analyze the intermediate quality of each module to guide the search. We remove this diagnostician since such intermediate quality is unavailable. \abacus conducts a bandit-style sequential search on sampled query subsets of $\cQ$, but its core mechanism requires knowing the model-wise quality of each module. To obtain this, in each iteration, we evaluate two configurations, where one is obtained from the other by changing the model choice in the module the algorithm is currently searching over. For random search, we iteratively evaluate a configuration randomly sampled from $\Theta$ without replacement on the entire dataset $\cQ$ until the search budget is exhausted.

\section{Ablation and Scalability Studies}\label{appendix:ablation}
In this section, we conduct ablation studies on key components and parameter settings of \ours and scalability evaluations on larger datasets.

\stitle{Ablation study}
We compare \ours with two variants, \ggours and \gours. 
\ggours replaces $\Theta_{\rm init}$ in Line~\ref{alg:scope-main:bc-sample-theta} of the \warmstart subroutine (Algorithm~\ref{alg:bound-calibration}) with a set of configurations sampled uniformly at random from $\Theta$, with the same size as $\Theta_{\rm init}$. 
\gours removes \warmstart and Line~\ref{alg:scope-main:break} from Algorithm~\ref{alg:scope-main}, and sets $t_0=0$. As a result, neither \ggours nor \gours requires $\theta_{\rm base}$ as input, and \gours evaluates the entire dataset in every iteration $i$ until the budget is exhausted.

Figure~\ref{fig:main-abl} reports the best feasible cost $c_{\rm bf}(\budget)$ versus the search budget $\budget$. We also conduct experiments varying other input parameters; across budgets and tasks, all three methods exhibit zero violations, and the output of \ours remains unchanged for $\alpha\in [1/5, 1/2-10^{-9}]$ and $\delta\in [10^{-5}, 10^{-3}]$.
As shown by \ggours, replacing $\Theta_{\rm init}$ in~\eqref{eq:theta-1-def} with a random configuration makes \ggours improve upon $\btheta_0$ more slowly than \ours at small budgets. For example, in Text-to-SQL, \ours reduces $c_{\rm bf}$ to $2.6\times10^{-3}$ by $\budget=5.5$, whereas \ggours first improves upon $\btheta_0$ only at $\budget=11.0$. This is because the randomly sampled configurations can be more expensive, leading to a higher initialization cost and delaying subsequent search progress in terms of budget. However, after surpassing $\btheta_0$, \ggours often converges faster than \ours. For instance, in data imputation, \ggours reduces the cost to $5.5\times 10^{-5}$ before $\budget=0.9$, even lower than \ours's best feasible cost at $\budget_{\rm max}$. To explain, a random configuration set can cover $\Theta$ more broadly, which may improve \warmstart over diverse configurations and help identify better candidates, although this advantage is unstable due to sampling randomness.
In contrast, \gours converges much more slowly than \ours. For instance, at $\budget_{\rm max}$, its best feasible cost is about $3-5\times$ higher than that of \ours, yet still lower than that of \BO, across tasks. While the guarantees in Theorem~\ref{thm:main-T-ours} can also apply to \gours if $Q=1$, \gours discards query-level signal exploitation and thus falls back to dataset-level evaluations, inheriting the cost inefficiency of existing competitors.

\begin{figure}[!t]
\captionsetup[subfloat]{labelformat=empty}
\centering
\begin{tikzpicture}
    \begin{customlegend}[legend columns=4,
        legend entries={\ours, \BO, \config, \llambo}
        ,
        legend style={at={(0.45,1.05)},anchor=north,draw=none,font=\footnotesize,column sep=0.1cm}]
        
    \addlegendimage{line width=0.35mm,mark size=3pt,mark=square*,color=RedPPAlt}
    \addlegendimage{line width=0.35mm,mark size=3pt,mark=o,color=Orange}
    \addlegendimage{line width=0.35mm,mark size=3pt,mark=diamond,color=LightGreen}
    \addlegendimage{line width=0.35mm,mark size=3pt,mark=triangle,color=Green}
    \end{customlegend}

    \begin{customlegend}[legend columns=4,
        legend entries={\abacus, \llmselect, \safeopt, \rand}
        ,
        legend style={at={(0.45,0.5)},anchor=north,draw=none,font=\footnotesize,column sep=0.1cm}]
        
    \addlegendimage{line width=0.35mm,mark size=3pt,mark=pentagon,color=LightBlue}
    \addlegendimage{line width=0.35mm,mark size=3pt,mark=star,color=Pink}
    \addlegendimage{line width=0.35mm,mark size=3pt,mark=triangle,mark options={rotate=180,solid,fill=white},color=DeepBlue}
    \addlegendimage{line width=0.35mm,mark size=3pt,mark=x,color=Gray}
    \end{customlegend}
\end{tikzpicture}
\\[-\lineskip]
\subfloat[]{
\begin{tikzpicture}[scale=1]
    \begin{axis}[
        axis line style={line width=0.3pt},
        height=\columnwidth/2.0,
        width=\columnwidth/1.4,
        ylabel={Best feasible cost},
        xmin=0, xmax=10.0,
        ymin=0.0001, ymax=0.01,
        xtick={0, 2, 4, 6, 8, 10},
        xticklabel style = {font=\footnotesize},
        xticklabels={0, 2, 4, 6, 8, 10},
        ymode=log,
        log basis y=10,
        grid=both,
major grid style={line width=0.3pt, draw=gray!35},
minor grid style={line width=0.2pt, draw=gray!20},
        ytick={0.0001,0.001,0.01},
        yticklabels={$10^{-4}$,$10^{-3}$,$10^{-2}$},
        every axis y label/.style={at={(current axis.north west)},anchor=south west, xshift=0mm, yshift=0mm},
        label style={font=\footnotesize},
        tick label style={font=\footnotesize},
after end axis/.code={
  \node[font=\footnotesize,anchor=north] at (rel axis cs:0.5,-0.09) {budget in USD};
},
        label style={font=\footnotesize},
        tick label style={font=\footnotesize},
    ]

\addplot[name path=SCOPEOptimizerUpper, draw=none] plot coordinates {
(0.000000, 0.005906)
(3.700985, 0.005906)
(3.700985, 0.000354)
(4.004720, 0.000354)
(4.004720, 0.000166)
(5.696888, 0.000166)
(5.696888, 0.000139)
(10.000000, 0.000139)
};
\addplot[name path=SCOPEOptimizerLower, draw=none] plot coordinates {
(0.000000, 0.005906)
(3.700985, 0.005906)
(3.700985, 0.000354)
(4.004720, 0.000354)
(4.004720, 0.000166)
(5.696888, 0.000166)
(5.696888, 0.000139)
(10.000000, 0.000139)
};
\addplot[fill=RedPPAlt, fill opacity=0.1, draw=none] fill between[of=SCOPEOptimizerUpper and SCOPEOptimizerLower];
\addplot[color=RedPPAlt, line width=0.5mm] plot coordinates {
(0.000000, 0.005906)
(3.700985, 0.005906)
(3.700985, 0.000354)
(4.004720, 0.000354)
(4.004720, 0.000166)
(5.696888, 0.000166)
(5.696888, 0.000139)
(10.000000, 0.000139)
};
\addplot[only marks, color=RedPPAlt, mark size=3.0pt, mark=square*, line width=0.5mm] coordinates {
(0.909091, 0.005906)
(1.818182, 0.005906)
(2.727273, 0.005906)
(3.636364, 0.005906)
(4.545455, 0.000166)
(5.454545, 0.000166)
(6.363636, 0.000139)
(7.272727, 0.000139)
(8.181818, 0.000139)
(9.090909, 0.000139)
};

\addplot[name path=BayesianOptimizerUpper, draw=none] plot coordinates {
(0.000000, 0.005906)
(9.913838, 0.005906)
(9.913838, 0.000801)
(10.000000, 0.000801)
};
\addplot[name path=BayesianOptimizerLower, draw=none] plot coordinates {
(0.000000, 0.005906)
(9.913838, 0.005906)
(9.913838, 0.000801)
(10.000000, 0.000801)
};
\addplot[fill=Orange, fill opacity=0.1, draw=none] fill between[of=BayesianOptimizerUpper and BayesianOptimizerLower];
\addplot[color=Orange, line width=0.5mm] plot coordinates {
(0.000000, 0.005906)
(9.913838, 0.005906)
(9.913838, 0.000801)
(10.000000, 0.000801)
};
\addplot[only marks, color=Orange, mark size=3.0pt, mark=o, line width=0.5mm] coordinates {
(0.909091, 0.005906)
(1.818182, 0.005906)
(2.727273, 0.005906)
(3.636364, 0.005906)
(4.545455, 0.005906)
(5.454545, 0.005906)
(6.363636, 0.005906)
(7.272727, 0.005906)
(8.181818, 0.005906)
(9.090909, 0.005906)
};

\addplot[name path=ConfigOptimizerUpper, draw=none] plot coordinates {
(0.000000, 0.005906)
(2.774097, 0.005906)
(2.774097, 0.000801)
(10.000000, 0.000801)
};
\addplot[name path=ConfigOptimizerLower, draw=none] plot coordinates {
(0.000000, 0.005906)
(2.774097, 0.005906)
(2.774097, 0.000801)
(10.000000, 0.000801)
};
\addplot[fill=LightGreen, fill opacity=0.1, draw=none] fill between[of=ConfigOptimizerUpper and ConfigOptimizerLower];
\addplot[color=LightGreen, line width=0.5mm] plot coordinates {
(0.000000, 0.005906)
(2.774097, 0.005906)
(2.774097, 0.000801)
(10.000000, 0.000801)
};
\addplot[only marks, color=LightGreen, mark size=3.0pt, mark=diamond, line width=0.5mm] coordinates {
(0.909091, 0.005906)
(1.818182, 0.005906)
(2.727273, 0.005906)
(3.636364, 0.000801)
(4.545455, 0.000801)
(5.454545, 0.000801)
(6.363636, 0.000801)
(7.272727, 0.000801)
(8.181818, 0.000801)
(9.090909, 0.000801)
};

\addplot[name path=SafeOptimizerUpper, draw=none] plot coordinates {
(0.000000, 0.005906)
(9.913838, 0.005906)
(9.913838, 0.000801)
(10.000000, 0.000801)
};
\addplot[name path=SafeOptimizerLower, draw=none] plot coordinates {
(0.000000, 0.005906)
(9.913838, 0.005906)
(9.913838, 0.000801)
(10.000000, 0.000801)
};
\addplot[fill=DeepBlue, fill opacity=0.1, draw=none] fill between[of=SafeOptimizerUpper and SafeOptimizerLower];
\addplot[color=DeepBlue, line width=0.5mm] plot coordinates {
(0.000000, 0.005906)
(9.913838, 0.005906)
(9.913838, 0.000801)
(10.000000, 0.000801)
};
\addplot[only marks, color=DeepBlue, mark size=3.0pt, mark=triangle, line width=0.5mm, mark options={rotate=180, solid, fill=white}] coordinates {
(0.909091, 0.005906)
(1.818182, 0.005906)
(2.727273, 0.005906)
(3.636364, 0.005906)
(4.545455, 0.005906)
(5.454545, 0.005906)
(6.363636, 0.005906)
(7.272727, 0.005906)
(8.181818, 0.005906)
(9.090909, 0.005906)
};

\addplot[name path=RandomSearcherUpper, draw=none] plot coordinates {
(0.000000, 0.005906)
(2.774097, 0.005906)
(2.774097, 0.000801)
(10.000000, 0.000801)
};
\addplot[name path=RandomSearcherLower, draw=none] plot coordinates {
(0.000000, 0.005906)
(2.774097, 0.005906)
(2.774097, 0.000801)
(10.000000, 0.000801)
};
\addplot[fill=Gray, fill opacity=0.1, draw=none] fill between[of=RandomSearcherUpper and RandomSearcherLower];
\addplot[color=Gray, line width=0.5mm] plot coordinates {
(0.000000, 0.005906)
(2.774097, 0.005906)
(2.774097, 0.000801)
(10.000000, 0.000801)
};
\addplot[only marks, color=Gray, mark size=3.0pt, mark=x, line width=0.5mm] coordinates {
(0.909091, 0.005906)
(1.818182, 0.005906)
(2.727273, 0.005906)
(3.636364, 0.000801)
(4.545455, 0.000801)
(5.454545, 0.000801)
(6.363636, 0.000801)
(7.272727, 0.000801)
(8.181818, 0.000801)
(9.090909, 0.000801)
};

\addplot[name path=LLAMBOOptimizerUpper, draw=none] plot coordinates {
(0.000000, 0.005906)
(2.841191, 0.005906)
(2.841191, 0.001239)
(5.375346, 0.001239)
(5.375346, 0.000279)
(10.000000, 0.000279)
};
\addplot[name path=LLAMBOOptimizerLower, draw=none] plot coordinates {
(0.000000, 0.005906)
(2.841191, 0.005906)
(2.841191, 0.001239)
(5.375346, 0.001239)
(5.375346, 0.000279)
(10.000000, 0.000279)
};
\addplot[fill=Green, fill opacity=0.1, draw=none] fill between[of=LLAMBOOptimizerUpper and LLAMBOOptimizerLower];
\addplot[color=Green, line width=0.5mm] plot coordinates {
(0.000000, 0.005906)
(2.841191, 0.005906)
(2.841191, 0.001239)
(5.375346, 0.001239)
(5.375346, 0.000279)
(10.000000, 0.000279)
};
\addplot[only marks, color=Green, mark size=3.0pt, mark=triangle, line width=0.5mm] coordinates {
(0.909091, 0.005906)
(1.818182, 0.005906)
(2.727273, 0.005906)
(3.636364, 0.001239)
(4.545455, 0.001239)
(5.454545, 0.000279)
(6.363636, 0.000279)
(7.272727, 0.000279)
(8.181818, 0.000279)
(9.090909, 0.000279)
};

\addplot[name path=LLMSelectorOptimizerUpper, draw=none] plot coordinates {
(0.000000, 0.005906)
(8.414596, 0.005906)
(8.414596, 0.003670)
(10.000000, 0.003670)
};
\addplot[name path=LLMSelectorOptimizerLower, draw=none] plot coordinates {
(0.000000, 0.005906)
(8.414596, 0.005906)
(8.414596, 0.003670)
(10.000000, 0.003670)
};
\addplot[fill=Pink, fill opacity=0.1, draw=none] fill between[of=LLMSelectorOptimizerUpper and LLMSelectorOptimizerLower];
\addplot[color=Pink, line width=0.5mm] plot coordinates {
(0.000000, 0.005906)
(8.414596, 0.005906)
(8.414596, 0.003670)
(10.000000, 0.003670)
};
\addplot[only marks, color=Pink, mark size=3.0pt, mark=star, line width=0.5mm] coordinates {
(0.909091, 0.005906)
(1.818182, 0.005906)
(2.727273, 0.005906)
(3.636364, 0.005906)
(4.545455, 0.005906)
(5.454545, 0.005906)
(6.363636, 0.005906)
(7.272727, 0.005906)
(8.181818, 0.005906)
(9.090909, 0.003670)
};

\addplot[name path=AbacusOptimizerUpper, draw=none] plot coordinates {
(0.000000, 0.005906)
(9.038211, 0.005906)
(9.038211, 0.001209)
(10.000000, 0.001209)
};
\addplot[name path=AbacusOptimizerLower, draw=none] plot coordinates {
(0.000000, 0.005906)
(9.038211, 0.005906)
(9.038211, 0.001209)
(10.000000, 0.001209)
};
\addplot[fill=LightBlue, fill opacity=0.1, draw=none] fill between[of=AbacusOptimizerUpper and AbacusOptimizerLower];
\addplot[color=LightBlue, line width=0.5mm] plot coordinates {
(0.000000, 0.005906)
(9.038211, 0.005906)
(9.038211, 0.001209)
(10.000000, 0.001209)
};
\addplot[only marks, color=LightBlue, mark size=3.0pt, mark=pentagon, line width=0.5mm] coordinates {
(0.909091, 0.005906)
(1.818182, 0.005906)
(2.727273, 0.005906)
(3.636364, 0.005906)
(4.545455, 0.005906)
(5.454545, 0.005906)
(6.363636, 0.005906)
(7.272727, 0.005906)
(8.181818, 0.005906)
(9.090909, 0.001209)
};

    \end{axis}
\end{tikzpicture}%
}%

\caption{Best feasible cost on entity resolution with 2293 queries.} \label{fig:scalability}
\vspace{-2mm}
\end{figure}
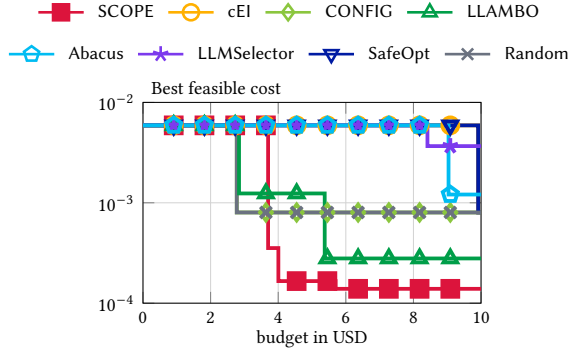

\stitle{Scalability evaluation}
To further validate \ours's performance, we experiment on an additional task, entity resolution, with a 3-module UniDM-ER system and dataset Amazon-google-dev of 2293 queries~\cite{qian2024unidm}. Figure~\ref{fig:scalability} shows the best feasible cost of each method over the search budget.
At the maximum budget, \ours is superior to other competitors in this new task, similar to the trends reflected in the original experiments. Specifically, the best feasible cost of \ours is 97\% lower than that of the reference configuration, and 50\% lower than that of the second-best method, \llambo.

\section{Proofs}\label{appendix:proof}

\subsection{Preliminaries of Proofs}

\stitle{Algorithm}
Recall that at time $t$, the observation history in an algorithm is $H_t=\{(\btheta_i, q_i, y_{c,i}, y_{g,i})\}_{i=1}^t$, with $H_0=\emptyset$. Let $\cF_t$ be the $\sigma$-algebra generated by $H_t$. We formally define an algorithm $\pi$ (e.g., \ours) as a tuple $(\tau, (\pi_t)_{t=1}^{\infty}, \btheta_{\rm{out}})$, where
\begin{itemize}[topsep=2pt,itemsep=1pt,parsep=0pt,partopsep=0pt,leftmargin=11pt]
    \item  $\tau$ is a stopping time with respect to the filtration $(\mathcal{F}_t)_{t=0}^{\infty}$;
    \item  $(\pi_t)_{t=1}^{\infty}$ is a sequence of $\mathcal{F}_{t-1}$-measurable mappings. For each $t \ge 1$, $\pi_t$ maps $H_{t-1}$ to a probability distribution over $\Theta\times \cQ$, denoted by $\pi_t(\cdot | H_{t-1})$. 
    At each time $t$ (for $1 \le t \leq \tau$), the pair $(\btheta_t, q_t)$ is sampled from $\pi_t(\cdot | H_{t-1})$, yielding observations $y_{c,t}$ and $y_{g,t}$;
    \item $\btheta_{\rm{out}}\in \Theta$ is an $\mathcal{F}_{\tau}$-measurable random variable representing the algorithm's recommendation of the best configuration.
\end{itemize}

\subsection{Proof of Theorem~\ref{cla:bounds}}

We first establish a query-wise confidence bound for each GP (one per query $q$ and metric $\zeta$),
and then aggregate across queries to obtain a confidence bound for the dataset-average metric. Recall $J_{q,t}=|\cJ_{q,t}|$ and let $j_q(1)<\cdots<j_q(J)$ be the times when query $q$ is evaluated up to time $t$.
For fixed $(q,\zeta,t)$, define $\btheta^{(a)}=\btheta_{j_q(a)}$ and
$y^{(a)}=y_{\zeta,j_q(a)}$ for $a\in[J]$, and let $J=J_{q,t}$.

We use the following standard kernelized self-normalized inequality (see, e.g., Theorem 1 and Lemma 1 of~\citep{chowdhury2017kernelized}).

\begin{lemma}[Kernelized self-normalization]\label{lem:kernel-self-normalized}
Let $(\btheta^{(a)},\eta^{(a)})_{a\ge 1}$ be a (possibly adaptive) sequence such that
$\eta^{(a)}$ is conditionally $R$-sub-Gaussian w.r.t.\ the filtration generated by the past.
For each $J\ge 1$, define $\Phi_J:\RR^J\to\cH_k$ by $\Phi_J e_a=\varphi(\btheta^{(a)})$,
$V_J=\Phi_J\Phi_J^\top+\lambda I_{\cH_k}$, and
$S_J=\sum_{a=1}^J \eta^{(a)}\varphi(\btheta^{(a)})$.
Then for any $\delta'\in(0,1)$, with probability at least $1-\delta'$,
simultaneously for all $J\ge 1$,
\[
\|S_J\|_{V_J^{-1}}
\le
R\sqrt{2\Big(\mathcal I_J+\log\frac{1}{\delta'}\Big)}
\le
R\sqrt{2\Big(\gamma(J)+\log\frac{1}{\delta'}\Big)},
\]
where $\mathcal I_J=\frac12\log\det(I+\lambda^{-1}K_J)$ and $K_J=\Phi_J^\top\Phi_J$.
\end{lemma}

Fix $q\in\cQ,~\zeta\in \{c,g\}$. By Assumption~\ref{ass:rkhs},
$\|f\|_{\mathcal H_k}\le B_\zeta$, where $\cH_k$ is the corresponding RKHS. Define $\Phi:\RR^J\to\mathcal H_k$ by $\Phi e_a=\varphi(\btheta^{(a)})$ and $V=\Phi\Phi^\top+\lambda I_{\mathcal H_k}$.
Let $\eta^{(a)}=y^{(a)}-f(\btheta^{(a)})$ and $\eta=[\eta^{(1)},\dots,\eta^{(J)}]^\top$.
By Assumption~\ref{ass:observation}, $(\eta^{(a)})$ are conditionally $R_\zeta$-sub-Gaussian.
As a standard property of the GP update in Definition~\ref{def:gp}, for any $\btheta\in\Theta$,
\begin{equation}\label{eq:krr-identity-correct}
f(\btheta)-\hat\mu_{\zeta,q,t}(\btheta)
=
\lambda\,\langle f, V^{-1}\varphi(\btheta)\rangle_{\mathcal H_k}
-
\langle S, V^{-1}\varphi(\btheta)\rangle_{\mathcal H_k},
\end{equation}
where $S=\sum_{a=1}^J \eta^{(a)}\varphi(\btheta^{(a)})=\Phi\eta$. Using the Woodbury matrix identity, one can verify the standard equality
\begin{equation}\label{eq:sigma-equals-correct}
\hat\sigma^2_{\zeta,q,t}(\btheta)
=
\lambda\,\langle \varphi(\btheta),V^{-1}\varphi(\btheta)\rangle_{\mathcal H_k}.
\end{equation}
Equivalently,
$\|V^{-1/2}\varphi(\btheta)\|_{\mathcal H_k}=\hat\sigma_{\zeta,q,t}(\btheta)/\sqrt{\lambda}$.
Next, the first term in~\eqref{eq:krr-identity-correct} is bounded by
\begin{align}
\left|\lambda\langle f,V^{-1}\varphi(\btheta)\rangle_{\mathcal H_k}\right|&=\lambda\left|\langle V^{-1/2}f,\,V^{-1/2}\varphi(\btheta)\rangle_{\mathcal H_k}\right|\notag\\
&\leq \sqrt{\lambda}\|f\|_{\mathcal H_k}\cdot \|V^{-1/2}\varphi(\btheta)\|_{\mathcal H_k}\notag\\
&\leq B_\zeta\,\hat\sigma_{\zeta,q,t}(\btheta)\label{eq:bias-final-correct},
\end{align}
where the equality is due to self-adjointness of $V^{-1/2}$, the first inequality uses Cauchy--Schwarz and $V\succeq \lambda I_{\mathcal H_k}$, and the second inequality uses \eqref{eq:sigma-equals-correct} and $\|f\|_{\mathcal H_k}\le B_\zeta$. Furthermore, the second term in ~\eqref{eq:krr-identity-correct} is bounded by
\begin{align}
\left|\langle S,V^{-1}\varphi(\btheta)\rangle_{\mathcal H_k}\right| &\leq \|S\|_{V^{-1}}\cdot \|V^{-1/2}\varphi(\btheta)\|_{\mathcal H_k}\notag\\
&=\frac{\|S\|_{V^{-1}}}{\sqrt{\lambda}}\,\hat\sigma_{\zeta,q,t}(\btheta)\notag\\
&\leq \frac{R_\zeta}{\sqrt{\lambda}}
\sqrt{2\Big(\gamma(J)+\log\frac{1}{\delta'}\Big)}
\cdot \hat\sigma_{\zeta,q,t}(\btheta)\label{eq:noise-final-correct},
\end{align}
where the first and last inequalities use Cauchy--Schwarz and Lemma~\ref{lem:kernel-self-normalized} with parameter $\delta'\in(0,1)$, respectively, and the equality is due to~\eqref{eq:sigma-equals-correct}. Combining \eqref{eq:krr-identity-correct}, \eqref{eq:bias-final-correct}, and \eqref{eq:noise-final-correct} yields
\[
\big|\hat\mu_{\zeta,q,t}(\btheta)-\ell_\zeta(\btheta,q)\big|
\le
\Bigg(
B_\zeta
+
\frac{R_\zeta}{\sqrt{\lambda}}
\sqrt{2\Big(\gamma(J)+\log\frac{1}{\delta'}\Big)}
\Bigg)\hat\sigma_{\zeta,q,t}(\btheta).
\]
Since $J=J_{q,t}\le J_{\max,t}$ and $\gamma(\cdot)$ is nondecreasing, $\gamma(J)\le \gamma(J_{\max,t})$, so the same inequality holds after replacing $\gamma(J)$ by $\gamma(J_{\max,t})$.

As there are $Q$ choices of $q$ and two choices of $\zeta$, we set $\delta'=\delta/(2Q)$ and use union bound. Then, with probability at least $1-\delta$, simultaneously for all $(t,q,\btheta,\zeta)$, we have
\[
\big|\hat\mu_{\zeta,q,t}(\btheta)-\ell_\zeta(\btheta,q)\big|
\!\le\!
\underbrace{\Bigg(\!
B_\zeta
+
\frac{R_\zeta}{\sqrt{\lambda}}
\sqrt{2\Big(\gamma(J_{\max,t})+\log\frac{2Q}{\delta}\Big)}
\Bigg)}_{\beta_{\zeta, t}/\sqrt{Q}}
\hat\sigma_{\zeta,q,t}(\btheta).
\]
Averaging over $q\in\cQ$ and applying Cauchy--Schwarz,
\begin{align*}
\big|\bar\mu_{\zeta,t}(\btheta)-\bar\ell_\zeta(\btheta)\big|
&=
\Big|\frac{1}{Q}\sum_{q\in\cQ}\big(\hat\mu_{\zeta,q,t}(\btheta)-\ell_\zeta(\btheta,q)\big)\Big|\\
&\le
\frac{\beta_{\zeta,t}}{\sqrt{Q}\,Q}\sum_{q\in\cQ}\hat\sigma_{\zeta,q,t}(\btheta)\\
&\le
\frac{\beta_{\zeta,t}}{\sqrt{Q}\,Q}\cdot \sqrt{Q}\,\sqrt{\sum_{q\in\cQ}\hat\sigma^2_{\zeta,q,t}(\btheta)}\\
&=
\beta_{\zeta,t}\sqrt{\sum_{q\in\cQ}\frac{\hat\sigma^2_{\zeta,q,t}(\btheta)}{Q^2}}
\ =\
\beta_{\zeta,t}\,\bar\sigma_{\zeta,t}(\btheta).
\end{align*}
Therefore, with probability at least $1-\delta$, simultaneous over
all $t\in\NN$, $q\in\cQ$, $\btheta\in\Theta$, and $\zeta\in\{c,g\}$,~\eqref{eq:lcb-ucb-property} follows immediately from the definitions of $L_{\zeta,t},U_{\zeta,t}$ in \eqref{eq:lcb-ours}. \qed

\subsection{Proof of Theorem~\ref{thm:main-T-ours}}

\stitle{Notations}
Let $\cE_{\mathrm{conf}}$ denote the high-probability event in~\eqref{eq:lcb-ucb-property}, i.e., all confidence intervals are valid.
Let $(\btheta_{\mathrm{out}},\tau)$ denote the output of Algorithm~\ref{alg:scope-main}, where $\btheta_{\mathrm{out}}$ is measurable w.r.t.\ the $\sigma$-algebra produced by the history
$\{(\btheta_t,q_t,y_{c,t},y_{g,t})\}_{t=1}^{\tau}$.
Accordingly, throughout the proof we analyze the execution restricted to times up to $\tau$.

Let $I$ be the number of executed rounds ($i$-loop iterations) whose candidate is evaluated at least once within times $\{t_0+1,\dots,\tau\}$,
indexed by $i\in[I]$ with strictly increasing start times
$t^{(1)}<\cdots<t^{(I)}$ and corresponding candidates $\btheta_{\mathrm{cand}}^{(i)}$.
Here, $t^{(i)}$ denotes the value of the global counter $t$ at the moment Line~\ref{alg:scope-main:select} is executed in round $i$ (before running the $j$-loop). Additionally, let $t^{(I+1)}=\tau$.
By construction, for each $i\in[I]$, we have $\btheta_{t'}=\btheta_{\mathrm{cand}}^{(i)}$ for all
$t'\in\{t^{(i)}+1, \dots, t^{(i+1)}\}$.

Define the set of certified-feasible rounds as
\[
\cI_{\mathrm{feas}}
\ =\
\Big\{
i\in[I]:\ \exists\, t\in[t^{(i)},\,t^{(i+1)}],\
U_{g,t}(\btheta_{\mathrm{cand}}^{(i)})\le 0
\Big\}.
\]
For each $i\in\cI_{\mathrm{feas}}$, define the set of certification-trigger times as
\[
\cT_{\mathrm{feas}}^{(i)}
\!=\!
\Big\{
t\in[t^{(i)}+1,\,t^{(i+1)}]\!:\!
\min\!\big\{
U_{g,t}(\btheta_{\mathrm{cand}}^{(i)}),\!
U_{g,t-1}(\btheta_{\mathrm{cand}}^{(i)})
\big\}\le 0
\Big\},
\]
and define
\[
t_{\mathrm{feas}}^{(i)}
\ \in\
\arg\min_{t\in\cT_{\mathrm{feas}}^{(i)}} U_{c,t}(\btheta_{\mathrm{cand}}^{(i)}),
\]
breaking ties arbitrarily.
Note that for every $i\in\cI_{\mathrm{feas}}$, $\cT_{\mathrm{feas}}^{(i)}\neq\emptyset$: if
$U_{g,t}(\btheta_{\mathrm{cand}}^{(i)})\le 0$ holds at some $t\in[t^{(i)},t^{(i+1)}]$, then taking
$t'=\max\{t,\,t^{(i)}+1\}\in[t^{(i)}+1,t^{(i+1)}]$ gives
$\min\{U_{g,t'}(\cdot),U_{g,t'-1}(\cdot)\}\le 0$, hence $t'\in\cT_{\mathrm{feas}}^{(i)}$.

\stitle{Proof of $\delta$-correctness}
Initially, $s(\btheta_{\mathrm{out}})=s(\btheta_0)\ge s_0$.
Whenever $\btheta_{\mathrm{out}}$ is updated in Line~\ref{alg:scope-main:update}, it is set to some candidate
$\btheta_{\mathrm{cand}}^{(i)}$ at a time $t$ such that
$\min\{U_{g,t}(\btheta_{\mathrm{cand}}^{(i)}),\,U_{g,t-1}(\btheta_{\mathrm{cand}}^{(i)})\}\le 0$.
Thus either $U_{g,t}(\btheta_{\mathrm{cand}}^{(i)})\le 0$ or $U_{g,t-1}(\btheta_{\mathrm{cand}}^{(i)})\le 0$.
On $\cE_{\mathrm{conf}}$, in either case we have $g(\btheta_{\mathrm{cand}}^{(i)})
\le
U_{g,t'}(\btheta_{\mathrm{cand}}^{(i)})
\le 0$ for some $t'\in\{t,t-1\}$,
hence $s(\btheta_{\mathrm{cand}}^{(i)})\ge s_0$.
Therefore $s(\btheta_{\mathrm{out}})\ge s_0$ holds on $\cE_{\mathrm{conf}}$, implying that \ours is $\delta$-correct.

\stitle{Proof outline for simple regret}
We focus on rounds where the selection threshold is small enough that $\btheta^\star$ satisfies the constraint in Line~\ref{alg:scope-main:select}. For these rounds, the selection rule guarantees that the chosen candidate has a cost lower confidence bound no larger than that of $\btheta^\star$, so $\mathrm{SR}$ can be controlled by the cost confidence width.
We then bound $\mathrm{SR}$ by the smallest such width among candidates that are eventually certified feasible, and replace this minimum by an average over the certified-feasible rounds, which introduces an inverse dependence on their count.
Finally, we lower bound the number of certified-feasible rounds: any round that is not certified feasible must exhibit sufficiently large constraint uncertainty relative to its threshold, while the total accumulated constraint uncertainty is bounded. Combining the two steps yields the stated $\mathrm{SR}$ bound.

Throughout the proof, we carry out all derivations on the event $\cE_{\mathrm{conf}}$; as shown above, $\Pr[\cE_{\mathrm{conf}}]\ge 1-\delta$.
Specifically, denote $\Delta^\star=s(\btheta^\star)-s_0=-g(\btheta^\star)>0$, and let $\cS=\{i^\star,i^\star+1,\dots,I\}\subseteq [I]$, where
\[
i^\star
\ =\
\min\{i\in\mathbb{N}:\ \Delta_i\le \Delta^\star\}
\ =\
\left\lceil\left({\Delta^\star}\right)^{-1/\alpha}\right\rceil. 
\]
Additionally, we assume that $\tau$ is sufficiently large, so that
\begin{equation}\label{eq:sr-assumption}
\cI_{\mathrm{feas}}\cap\cS\neq \emptyset.\tag{$\spadesuit$}
\end{equation}
In particular, this implies that $\cI_{\mathrm{feas}}\neq \emptyset$ and $\cS\neq \emptyset$. We will show at the end of the proof that~\eqref{eq:sr-assumption} is implied by the condition under which~\eqref{eq:sr-bound-T} is stated. Additionally, we will use the following lemma.

\begin{lemma}\label{lem:sum-variance-bound}
For any $\zeta\in\{c,g\}$,
\[
\sum_{t=t_0+1}^{\tau}\bar\sigma_{\zeta,t-1}^2(\btheta_t)
\ \le\
4(\lambda+1)\gamma(J_{\max,\tau}).
\]
\end{lemma}
Since $J_{\max,\tau}\le \tau$ and $\gamma(\cdot)$ is nondecreasing, we have $\gamma(J_{\max,\tau})\ \le\ \gamma(\tau).$
Accordingly, in the sequel we may upper bound any occurrence of $\gamma(J_{\max,\tau})$ by $\gamma(\tau)$.

\stitle{Relating $\btheta_{\rm out}$ to $\btheta_{\rm cand}$}
Let $U_{\rm out}$ denote the final value of the variable $U_{\rm out}$ upon termination of
Algorithm~\ref{alg:scope-main}. By construction, $U_{\rm out}$ is initialized as $U_{c,t_0}(\btheta_0)$
and is updated only when Line~\ref{alg:scope-main:update} holds; hence $U_{\rm out}$ is nonincreasing over time.
Moreover, whenever round $i$ triggers the modified update condition at some time
$t\in\cT_{\mathrm{feas}}^{(i)}$, the algorithm compares $U_{c,t}(\btheta_{\mathrm{cand}}^{(i)})$ with the current
$U_{\rm out}$ and updates $U_{\rm out}$ if this value is smaller. Therefore, at termination,
\begin{equation}\label{eq:Uout-min-certified}
U_{\rm out}
\ \le\
U_{c,t_{\mathrm{feas}}^{(i)}}(\btheta_{\mathrm{cand}}^{(i)}),
\qquad \forall i\in\cI_{\mathrm{feas}}.
\end{equation}

Fix any $i\in\cI_{\mathrm{feas}}$. On $\cE_{\mathrm{conf}}$, we have
\begin{align}
    c(\btheta_{\mathrm{out}})
    &\leq U_{\rm out} \notag\\
    &\leq U_{c,t_{\mathrm{feas}}^{(i)}}(\btheta_{\mathrm{cand}}^{(i)}) \tag{using~\eqref{eq:Uout-min-certified}}\\
    &=L_{c,t_{\mathrm{feas}}^{(i)}}(\btheta_{\mathrm{cand}}^{(i)})+
2\beta_{c,t_{\mathrm{feas}}^{(i)}}\bar\sigma_{c,t_{\mathrm{feas}}^{(i)}}(\btheta_{\mathrm{cand}}^{(i)}) \notag\\
    &\leq c(\btheta_{\mathrm{cand}}^{(i)})+
2\beta_{c,t_{\mathrm{feas}}^{(i)}}\bar\sigma_{c,t_{\mathrm{feas}}^{(i)}}(\btheta_{\mathrm{cand}}^{(i)}) \notag\\
    &\leq c(\btheta_{\mathrm{cand}}^{(i)})+
2\beta_{c,t_{\mathrm{feas}}^{(i)}}\bar\sigma_{c,t^{(i)}}(\btheta_{\mathrm{cand}}^{(i)}) \tag{$\bar \sigma$ is nonincreasing}\\
    &\leq c(\btheta_{\mathrm{cand}}^{(i)})+
2\beta_{c,\tau}\bar\sigma_{c,t^{(i)}}(\btheta_{\mathrm{cand}}^{(i)}),\label{eq:cout-bound}
\end{align}
where the last inequality uses that $\beta_{c,t}$ is nondecreasing in $t$.

\stitle{Relating to $\btheta^{\star}$}
By definition of $\cS$, on $\cE_{\rm conf}$, we have
\[
L_{g,t^{(i)}}(\btheta^\star) \le -\Delta^\star \leq
-\Delta_i,
\qquad \forall i\in\cS.
\]
In other words, every index $i\in\cS$ has sufficiently small $\Delta_i$, so that $\btheta^\star$ {\it can} be selected by Line~\ref{alg:scope-main:select} at time $t^{(i)}$. By the selection rule, for any $i \in \cI_{\rm feas}\cap \cS$, we have
\[
L_{c,t^{(i)}}(\btheta_{\mathrm{cand}}^{(i)})
\ \le\
L_{c,t^{(i)}}(\btheta^\star).
\]
On $\cE_{\mathrm{conf}}$, we further bound $c(\btheta_{\mathrm{cand}}^{(i)})-c(\btheta^\star)$ as follows.
\begin{align}
c(\btheta_{\mathrm{cand}}^{(i)})-c(\btheta^\star)
&\le
U_{c,t^{(i)}}(\btheta_{\mathrm{cand}}^{(i)})-L_{c,t^{(i)}}(\btheta^\star)\nonumber\\
&\le
U_{c,t^{(i)}}(\btheta_{\mathrm{cand}}^{(i)})-L_{c,t^{(i)}}(\btheta_{\mathrm{cand}}^{(i)})\nonumber\\
&= 2\beta_{c,t^{(i)}}\bar\sigma_{c,t^{(i)}}(\btheta_{\mathrm{cand}}^{(i)})\nonumber\\
&\leq 2\beta_{c,\tau}\bar\sigma_{c,t^{(i)}}(\btheta_{\mathrm{cand}}^{(i)}).\label{eq:cand-gap}
\end{align}
Combining~\eqref{eq:cout-bound} and~\eqref{eq:cand-gap} yields that, for every
$i\in\cI_{\mathrm{feas}}\cap\cS$,
\[
c(\btheta_{\mathrm{out}})-c(\btheta^\star)
\ \le\
4\beta_{c,\tau}\bar\sigma_{c,t^{(i)}}(\btheta_{\mathrm{cand}}^{(i)}).
\]
Taking the minimum over $i$ yields
\begin{align}
c(\btheta_{\mathrm{out}})-c(\btheta^\star)&\le
4\beta_{c,\tau}\cdot
\min_{i\in \cI_{\mathrm{feas}}\cap\cS}\bar\sigma_{c,t^{(i)}}(\btheta_{\mathrm{cand}}^{(i)})\nonumber\\
&\leq 4\beta_{c,\tau}\cdot \sqrt{ \frac{\sum_{i\in \cI_\mathrm{feas}\cap\cS}\bar\sigma_{c,t^{(i)}}^2(\btheta_{\mathrm{cand}}^{(i)})}{|\cI_\mathrm{feas}\cap\cS|}}\nonumber\\
&\leq 4\beta_{c,\tau}\cdot \sqrt{ \frac{\sum_{t=t_0+1}^{\tau}\bar\sigma_{c,t-1}^2(\btheta_t)}{|\cI_\mathrm{feas}\cap\cS|}}\nonumber\\
&\leq 4\beta_{c,\tau}\cdot \sqrt{\frac{4(\lambda+1)\gamma(J_{\max,\tau})}{|\cI_\mathrm{feas}\cap\cS|}}\nonumber\\
&\leq 4\beta_{c,\tau}\cdot \sqrt{\frac{4(\lambda+1)\gamma(\tau)}{|\cI_\mathrm{feas}\cap\cS|}}
\nonumber\\
&= \frac{\mathsf{A}}{\sqrt{Q}\cdot\sqrt{|\cI_\mathrm{feas}\cap\cS|}},\label{eq:sr-bound-with-size}
\end{align}
where $\mathsf A=8\beta_{c,\tau}\sqrt{Q(\lambda+1)\gamma(\tau)}$.
Here, the second inequality replaces the minimum with the average, the third inequality uses the fact that for each $i\in[I]$, the sum $\sum_{t=t_0+1}^{\tau}\bar\sigma_{c,t-1}^2(\btheta_t)$ contains the term
\[
\bar\sigma_{c,t^{(i)}}^2(\btheta_{t^{(i)}+1})
=
\bar\sigma_{c,t^{(i)}}^2(\btheta_{\mathrm{cand}}^{(i)}),
\]
and the fourth inequality uses Lemma~\ref{lem:sum-variance-bound} with $\zeta=c$. 

\stitle{Bounding $|\cI_\mathrm{feas}\cap\cS|$}
Let $\cI_{\mathrm{bad}}^{\mathrm{feas}}=[I]\setminus\cI_{\mathrm{feas}}$, $I_{\rm bad} = |\cI_{\mathrm{bad}}^{\mathrm{feas}}|$, and fix index $i\in \cI_{\mathrm{bad}}^{\mathrm{feas}}$.
On one hand, since round $i$ starts, the selection rule implies
$L_{g,t^{(i)}}(\btheta_{\mathrm{cand}}^{(i)})\le -\Delta_i$ with $\Delta_i=i^{-\alpha}$.
On the other hand, because $i\notin\cI_{\mathrm{feas}}$, by definition we have
\[
U_{g,t}(\btheta_{\mathrm{cand}}^{(i)})>0,
\qquad \forall t\in[t^{(i)},t^{(i+1)}].
\]
In particular, $U_{g,t^{(i)}}(\btheta_{\mathrm{cand}}^{(i)})>0$. Therefore,
\begin{align*}
0
&<
U_{g,t^{(i)}}(\btheta_{\mathrm{cand}}^{(i)})\\
&= L_{g,t^{(i)}}(\btheta_{\mathrm{cand}}^{(i)}) + 2\beta_{g,t^{(i)}}\bar\sigma_{g,t^{(i)}}(\btheta_{\mathrm{cand}}^{(i)})\\
&\leq -\Delta_i
+
2\beta_{g,t^{(i)}}\bar\sigma_{g,t^{(i)}}(\btheta_{\mathrm{cand}}^{(i)})\\
&\leq -\Delta_i
+
2\beta_{g,\tau}\bar\sigma_{g,t^{(i)}}(\btheta_{\mathrm{cand}}^{(i)})\\
&\Leftrightarrow\ 
2\beta_{g,\tau}\cdot \bar\sigma_{g,t^{(i)}}(\btheta_{\mathrm{cand}}^{(i)})>\Delta_i.
\end{align*}
Squaring and summing over $i\in\cI_{\mathrm{bad}}^{\mathrm{feas}}$ gives
\[
\sum_{i\in\cI_{\mathrm{bad}}^{\mathrm{feas}}}\Delta_i^2
\ \le\
4\beta_{g,\tau}^2
\sum_{i\in\cI_{\mathrm{bad}}^{\mathrm{feas}}}\bar\sigma_{g,t^{(i)}}^2(\btheta_{\mathrm{cand}}^{(i)}).
\]
As each term
$\bar\sigma_{g,t^{(i)}}^2(\btheta_{\mathrm{cand}}^{(i)})$ appears in the sum
$\sum_{t=t_0+1}^{\tau}\bar\sigma_{g,t-1}^2(\btheta_t)$, we apply Lemma~\ref{lem:sum-variance-bound} with $\zeta=g$, which yields
\begin{equation}\label{eq:sum-delta2-ub-alpha}
\sum_{i\in\cI_{\mathrm{bad}}^{\mathrm{feas}}}\Delta_i^2
\ \le\
16(\lambda+1)\beta_{g,\tau}^2\gamma(J_{\max,\tau}).
\end{equation}

For the bound on $I_{\rm bad}$, since $\Delta_i=i^{-\alpha}$, we have $\Delta_i^2=i^{-2\alpha}$, which is nonincreasing in $i$.
Therefore, the minimum of $\sum_{i\in S}\Delta_i^2$ over all $S\subseteq[I]$ with $|S|=I_{\rm bad}$ is achieved at
$S=\{I-I_{\rm bad}+1,\dots,I\}$. It follows that
\begin{align*}
\sum_{i\in \cI_{\mathrm{bad}}^{\mathrm{feas}}}\Delta_i^2
&\ge
\sum_{j=I-I_{\rm bad}+1}^I j^{-2\alpha}\\
&\ge
I_{\rm bad}\cdot I^{-2\alpha},
\end{align*}
where the last inequality uses $j\le I$ for all $j\in\{I-I_{\rm bad}+1,\dots,I\}$.
Combining this with~\eqref{eq:sum-delta2-ub-alpha} yields
\begin{equation}\label{eq:bad-count-alpha}
I_{\rm bad}
\ \le\
16(\lambda+1)\beta_{g,\tau}^2\gamma(J_{\max,\tau})\cdot I^{2\alpha}.
\end{equation}

Now we derive a lower bound of $|\cI_{\mathrm{feas}}\cap\cS|$.
Recall that $\cS=\{i^\star,i^\star+1,\dots,I\}$ and thus $|\cS|=I-i^\star+1$. Therefore,
\begin{align*}
|\cI_{\mathrm{feas}}\cap\cS|
&=
|\cS|-|\cI_{\mathrm{bad}}^{\mathrm{feas}}\cap\cS|\\
& \geq (I-i^\star+1)-I_{\rm bad}\\
& \geq (I-i^\star+1)-
16(\lambda+1)\beta_{g,\tau}^2\gamma(J_{\max,\tau})\cdot I^{2\alpha}.
\end{align*}
Next, relate $I$ and $\tau-t_0$.
Each effective round evaluates at most $Q$ distinct queries, so $\tau-t_0 \le QI$ and thus
$I\ge (\tau-t_0)/Q$.
Additionally, each effective round records at least one observation, so $I\le \tau-t_0$ and hence
$I^{2\alpha}\le (\tau-t_0)^{2\alpha}$.
Substituting these yields
\begin{align}
|\cI_{\mathrm{feas}}\cap\cS|
&\ge \frac{\tau-t_0}{Q}-i^\star
-16(\lambda+1)\beta_{g,\tau}^2\gamma(J_{\max,\tau})\cdot (\tau-t_0)^{2\alpha}\nonumber\\
&\geq \frac{(\tau-t_0)-\mathsf B\cdot(\tau-t_0)^{2\alpha}-\mathsf C}{Q},
\label{eq:bound-feas-size-alpha}
\end{align}
where $\mathsf B
=
16Q(\lambda+1)\beta_{g,\tau}^2\gamma(J_{\tau})$ and $\mathsf C=Q\cdot i^\star$.

\stitle{Conclusion}
Plugging~\eqref{eq:bound-feas-size-alpha} into~\eqref{eq:sr-bound-with-size} gives
\[
\mathrm{SR}(\tau)
\ \le\
\frac{\mathsf{A}}{\sqrt{(\tau-t_0)-\mathsf B\cdot(\tau-t_0)^{2\alpha}-\mathsf C}},
\]
which is well-defined whenever $(\tau-t_0)-\mathsf B\cdot(\tau-t_0)^{2\alpha}-\mathsf C>0$.
By~\eqref{eq:bound-feas-size-alpha}, this condition implies $|\cI_{\mathrm{feas}}\cap\cS|>0$, i.e., \eqref{eq:sr-assumption} holds.
This completes the proof. \qed

\begin{proof}[Proof of Lemma~\ref{lem:sum-variance-bound}]
Fix $\zeta\in\{c,g\}$, and recall that $\cJ_{q,t}=\{j\in[t]:q_j=q\}=\{j_q(1),\cdots,j_q(J_{q,t})\}$ with $J_{q,t}=|\cJ_{q,t}|$ and $j_q(1)<\cdots<j_q(J_{q,t})$. We denote the information gain of a configuration set $A=\{\btheta^{(1)},\dots,\btheta^{(J)}\}\subseteq \Theta$, in any order, by $\cI(A)
\ =\
\frac{1}{2}\log\det(\bI+\lambda^{-1}\bK_A)$.

\stitle{Bounding the sum of query-wise variances}
Fix a query $q\in\cQ$ and define the evaluation sequence
$\btheta_{q,r}=\btheta_{j_q(r)}$ for $r\in[J_{q,\tau}]$.
Let
\[
v_{q,r}
\ =\
\hat\sigma^2_{\zeta,q,\,j_q(r)-1}(\btheta_{q,r}),
\qquad r\in[J_{q,\tau}].
\]
For the observed configurations $\{\btheta_{q,1},\dots,\btheta_{q,J_{q,\tau}}\}$ on query $q$, a standard GP determinant identity yields
\[
\mathcal I(\{\btheta_{q,1},\dots,\btheta_{q,J_{q,\tau}}\})
\ =\
\frac12\sum_{r=1}^{J_{q,\tau}}\log\!\Big(1+v_{q,r}/\lambda\Big).
\]
Since the set has size $J_{q,\tau}$, we have
$\mathcal I(\cdot)\le \gamma(J_{q,\tau})\le \gamma(J_{\max,\tau})$, and thus
\begin{equation}\label{eq:ig-sumlog}
\sum_{r=1}^{J_{q,\tau}}\log\!\Big(1+v_{q,r}/\lambda\Big)
\ \le\
2\gamma(J_{\max,\tau}).
\end{equation}
Moreover, for any $x\in[0,1]$,
$\log(1+x/\lambda)\ge x/(\lambda+x)\ge x/(\lambda+1)$, hence
\begin{equation}\label{eq:x-log}
x \ \le\ (\lambda+1)\log(1+x/\lambda).
\end{equation}
Applying \eqref{eq:x-log} to $x=v_{q,r}$ and summing, then using \eqref{eq:ig-sumlog}, gives
\[
\sum_{r=1}^{J_{q,\tau}} v_{q,r}
\ \le\
(\lambda+1)\sum_{r=1}^{J_{q,\tau}}\log(1+v_{q,r}/\lambda)
\ \le\
2(\lambda+1)\gamma(J_{\max,\tau}).
\]
Summing over $q\in\cQ$ and reindexing by time $t$ yields
\begin{equation}\label{eq:sigma-sum-by-ig-clean}
\sum_{t=1}^{\tau}\hat\sigma^2_{\zeta,q_t,t-1}(\btheta_t)
\ \le\
2(\lambda+1)\,Q\,\gamma(J_{\max,\tau}).
\end{equation}

\stitle{Bounding the aggregated uncertainty sum}
Consider any fixed inner iteration (Lines~\ref{alg:scope-main:for}--\ref{alg:scope-main:update}) in which the algorithm
keeps $\btheta_t=\btheta_{\rm cand}$ for $m\le Q$ consecutive times and selects $m$ distinct queries.
Let the corresponding times be $u_1<\cdots<u_m$ and define
\[
v_j
\ =\
\hat\sigma^2_{\zeta,q_{u_j},u_j-1}(\btheta_{\rm cand}),
\qquad j\in[m].
\]
At time $u_j-1$, for any query $q$ not yet selected in this inner iteration, its posterior variance equals its value at
the iteration start, because no observation of query $q$ has been made.
Moreover, by the pre-ordering of $q_{u_j}$ via $\varphi_i$ in~\eqref{eq:uncertainty-score}, this variance is at most $v_j$.
For any previously selected query $q_{u_r}$ with $r<j$, monotonicity of $\hat\sigma_{\cdot,q,\cdot}$ with respect to time
(implied by the GP update rules) implies its variance at time $u_j-1$ is at most $v_r$ (which is the variance at time
$u_r-1$ before observing $q_{u_r}$ in this inner iteration). Hence,
\[
\sum_{q\in\cQ}\hat\sigma^2_{\zeta,q,u_j-1}(\btheta_{\rm cand})
\ \le\
\sum_{r=1}^{j-1} v_r\ +\ (Q-j+1)\,v_j.
\]
Dividing by $Q^2$ gives an upper bound on $\bar\sigma^2_{\zeta,u_j-1}(\btheta_{\rm cand})$.
Summing over $j=1,\dots,m$ and collecting coefficients yields
\begin{align*}
\sum_{j=1}^m \bar\sigma^2_{\zeta,u_j-1}(\btheta_{\rm cand})
&\le \frac{1}{Q^2}\sum_{r=1}^m\big((m-r)+(Q-r+1)\big)v_r\\
&\le \frac{Q+m}{Q^2}\sum_{r=1}^m v_r\\
&\leq \frac{2}{Q}\sum_{r=1}^m v_r,
\end{align*}
where the last inequality uses $m\le Q$.
Summing the above bound over all inner iterations, and noting that
$\sum_{r=1}^m v_r$ over all inner iterations equals $\sum_{t=t_0+1}^{\tau}\hat\sigma^2_{\zeta,q_t,t-1}(\btheta_t)$,
we obtain
\begin{equation}\label{eq:barsigma-to-hatsigma}
\sum_{t=t_0+1}^{\tau}\bar\sigma^2_{\zeta,t-1}(\btheta_t)
\ \le\
\frac{2}{Q}\sum_{t=t_0+1}^{\tau}\hat\sigma^2_{\zeta,q_t,t-1}(\btheta_t)
\ \le\
\frac{2}{Q}\sum_{t=1}^{\tau}\hat\sigma^2_{\zeta,q_t,t-1}(\btheta_t).
\end{equation}
Combining \eqref{eq:barsigma-to-hatsigma} with \eqref{eq:sigma-sum-by-ig-clean} yields
\[
\sum_{t=t_0+1}^{\tau}\bar\sigma^2_{\zeta,t-1}(\btheta_t)
\ \le\
4(\lambda+1)\gamma(J_{\max,\tau}).
\]
This completes the proof.
\end{proof}

\subsection{Proof of Corollary~\ref{cor:exp-sr-budget}}

We first derive a bound that is explicit in both $\budget$ and a failure probability $\delta\in(0,1)$,
and then plug in $\delta=\budget^{-2}$.

Let $\Delta_C=C_{\max}-C_{\min}$. Since $\ell_c(\btheta,q)\in[C_{\min},C_{\max}]$, we have
$c(\btheta)\in[C_{\min},C_{\max}]$ for all $\btheta$, and hence $\mathrm{SR}(\tau)\le \Delta_C$ always.
Fix $\budget>0$ and $\delta\in(0,1)$, and run \ours with failure probability $\delta$,
where $\beta_{\zeta,t}$ is set equal to the right-hand side of~\eqref{eq:beta-def}.

Let $\cE_{\mathrm{conf}}$ be the confidence event used in the proof of Theorem~\ref{thm:main-T-ours}.
As established in that proof, the bound~\eqref{eq:sr-bound-T} holds on $\cE_{\mathrm{conf}}$, and
\begin{equation}\label{eq:sr-prob-app}
\Pr[\cE_{\mathrm{conf}}]\ \ge\ 1-\delta.
\end{equation}

Because $\Theta$ is finite, $\gamma(J_{\max,t})\le \gamma(|\Theta|)$ and $\gamma(t)\le \gamma(|\Theta|)$ for all $t$.
Since $\beta_{\zeta,t}$ is set equal to the right-hand side of~\eqref{eq:beta-def}, for all $t\in\NN$ and $\zeta\in\{c,g\}$, $\beta_{\zeta,t} \leq \bar\beta_{\zeta}(\delta)$,
where $\bar\beta_{\zeta}(\delta)$ is defined by replacing $\gamma(J_{\max, t})$ in~\eqref{eq:beta-def} with $\gamma(|\Theta|)$.
Since $\gamma(\tau)\le \gamma(|\Theta|)$ and $\beta_{\zeta,\tau}\le \bar\beta_{\zeta}(\delta)$, the constants $\mathsf A,\mathsf B,\mathsf C$
in Theorem~\ref{thm:main-T-ours} also satisfy the deterministic upper bounds
\[
\mathsf A\ \le\ \bar{\mathsf A}(\delta),\qquad
\mathsf B\ \le\ \bar{\mathsf B}(\delta),\qquad
\mathsf C\ =\ Q\left\lceil(-g(\btheta^\star))^{-1/\alpha}\right\rceil,
\]
where
\[
\bar{\mathsf A}(\delta)
=
8\,\bar\beta_c(\delta)\,\sqrt{Q(\lambda+1)\gamma(|\Theta|)},
\quad
\bar{\mathsf B}(\delta)
=
16Q(\lambda+1)\,\bar\beta_g(\delta)^2\,\gamma(|\Theta|).
\]

We define an explicit $m_1(\delta)$ that ensures the denominator in~\eqref{eq:sr-bound-T} is at least $(\tau-t_0)/2$.
Let $m_1(\delta)$ be the smallest integer satisfying
\begin{equation}\label{eq:def-m1-app}
m_1(\delta)\ \ge\ \max\!\left\{\left\lceil(4\bar{\mathsf B}(\delta))^{\frac{1}{1-2\alpha}}\right\rceil,\ 4\mathsf C\right\}.
\end{equation}
Then for any integer $m\ge m_1(\delta)$ we have $\bar{\mathsf B}(\delta)\,m^{2\alpha}\le m/4$ and $\mathsf C\le m/4$, and hence
\begin{equation}\label{eq:denom-lb-app}
m-\mathsf B\,m^{2\alpha}-\mathsf C
\ \ge\
m-\bar{\mathsf B}(\delta)\,m^{2\alpha}-\mathsf C
\ \ge\
\frac{m}{2}.
\end{equation}
Therefore, on $\cE_{\mathrm{conf}}\cap\{\tau-t_0\ge m_1(\delta)\}$,
\begin{align}
\mathrm{SR}(\tau)
&\le \frac{\mathsf A}{\sqrt{(\tau-t_0)-\mathsf B(\tau-t_0)^{2\alpha}-\mathsf C}}
\tag{using Theorem~\ref{thm:main-T-ours}} \notag\\
&\le \frac{\bar{\mathsf A}(\delta)}{\sqrt{(\tau-t_0)-\mathsf B(\tau-t_0)^{2\alpha}-\mathsf C}}
\tag{since $\mathsf A\le \bar{\mathsf A}(\delta)$} \notag\\
&\le \frac{\bar{\mathsf A}(\delta)}{\sqrt{(\tau-t_0)/2}}
\tag{using~\eqref{eq:denom-lb-app} with $m=\tau-t_0$} \notag\\
&\le \frac{\sqrt{2}\,\bar{\mathsf A}(\delta)}{\sqrt{\tau-t_0}}.
\label{eq:sr-vs-tau-app}
\end{align}

Next, we control the probability that the budget-based stopping time $\tau$ is small.
Recall $y_{c,t}=\ell_c(\btheta_t,q_t)+\eta_{c,t}$ and let $M_s=\sum_{t=1}^s \eta_{c,t}$. If $R_c=0$, then $M_s=0$ almost surely for all $s$, and the tail probability below is zero whenever $nC_{\max}<\budget$. Thus it remains to prove the tail bound when $R_c>0$.

\begin{lemma}\label{lem:max-ineq-cond-gauss}
Assume $R_c>0$. For any integer $n\ge 1$ and any $x>0$,
\[
\Pr\!\left[\max_{1\le s\le n} M_s>x\right]\ \le\ \exp\!\left(-\frac{x^2}{2nR_c^2}\right).
\]
\end{lemma}
\noindent

Since $\ell_c(\btheta_t,q_t)\le C_{\max}$, for every $s$,
\begin{equation}\label{eq:sumy-bound-app}
\sum_{t=1}^s y_{c,t}
=
\sum_{t=1}^s \ell_c(\btheta_t,q_t)+\sum_{t=1}^s \eta_{c,t}
\le sC_{\max}+M_s.
\end{equation}
By Lemma~\ref{lem:max-ineq-cond-gauss}, for any integer $n\ge 1$ such that $\budget>nC_{\max}$,
\begin{align}
\Pr[\tau\le n]
&\le \Pr\!\left[\exists s\le n:\ \sum_{t=1}^s y_{c,t}>\budget\right]\notag\\
&\le \Pr\!\left[\max_{1\le s\le n} M_s>\budget-nC_{\max}\right]\notag\\
&\le \exp\!\left(-\frac{(\budget-nC_{\max})^2}{2nR_c^2}\right).
\label{eq:tau-tail-app}
\end{align}
Choose $n=\lfloor \budget/(2C_{\max})\rfloor$. Then $\budget-nC_{\max}\ge \budget/2$ and
$n\le \budget/(2C_{\max})$, so~\eqref{eq:tau-tail-app} yields
\begin{equation}\label{eq:tau-tail-simple-app}
\Pr[\tau\le n]
\ \le\
\exp\!\left(-\frac{\budget\,C_{\max}}{4R_c^2}\right).
\end{equation}

Assume in addition that $\budget\ge 4C_{\max}(t_0+1)$. Then $n=\lfloor \budget/(2C_{\max})\rfloor$ satisfies
$n-t_0\ge \budget/(4C_{\max})$.
If also $n-t_0\ge m_1(\delta)$, then on $\cE_{\mathrm{conf}}\cap\{\tau>n\}$ we have $\tau-t_0\ge n-t_0\ge m_1(\delta)$, and hence
\begin{align}
\mathrm{SR}(\tau)
&\le \frac{\sqrt{2}\,\bar{\mathsf A}(\delta)}{\sqrt{\tau-t_0}}
\tag{using~\eqref{eq:sr-vs-tau-app}}\notag\\
&\le \frac{\sqrt{2}\,\bar{\mathsf A}(\delta)}{\sqrt{n-t_0}}
\tag{since $\tau>n$}\notag\\
&\le \frac{2\sqrt{C_{\max}}\,\bar{\mathsf A}(\delta)}{\sqrt{\budget}}.
\tag{since $n-t_0\ge \budget/(4C_{\max})$}\label{eq:sr-on-taun-app}
\end{align}

Taking expectation and using $\mathrm{SR}(\tau)\le \Delta_C$, we have $\EE[\mathrm{SR}(\tau)]$
\begin{align*}
=&
\EE\!\left[\mathrm{SR}(\tau)\,\bm 1\{\cE_{\mathrm{conf}}\cap(\tau>n)\}\right]
+\EE\!\left[\mathrm{SR}(\tau)\,\bm 1\{\cE_{\mathrm{conf}}^c\cup(\tau\le n)\}\right]\\
\le&
\frac{2\sqrt{C_{\max}}\,\bar{\mathsf A}(\delta)}{\sqrt{\budget}}
+\Delta_C\,\Pr[\cE_{\mathrm{conf}}^c]
+\Delta_C\,\Pr[\tau\le n]\\
\le&
\frac{2\sqrt{C_{\max}}\,\bar{\mathsf A}(\delta)}{\sqrt{\budget}}
+\Delta_C\,\delta
+\Delta_C\,\exp\!\left(-\frac{\budget\,C_{\max}}{4R_c^2}\right),
\end{align*}
where the last inequality uses~\eqref{eq:sr-prob-app} and~\eqref{eq:tau-tail-simple-app}. When $R_c=0$, the exponential term is omitted.
This gives an explicit bound in terms of $(\budget,\delta)$, provided that
$\budget\ge 4C_{\max}(t_0+1)$ and $n-t_0\ge m_1(\delta)$. Now set $\delta=\budget^{-2}$ and restrict to $\budget\ge 2$ so that $\delta\in(0,1)$.
For all $\budget\ge 2$, $\log(2Q/\delta)\le \log(2Q)+2\log\budget$.
Thus there exists a constant $K_1>0$, independent of $\budget$, such that for all $\budget\ge 2$,
\[
\bar{\mathsf B}(\budget^{-2})\ \le\ K_1(1+\log \budget).
\]
Therefore, by~\eqref{eq:def-m1-app}, there exist constants $K_2,K_3>0$, independent of $\budget$, such that for all $\budget\ge 2$,
\begin{equation}\label{eq:m1-polylog-app}
m_1(\budget^{-2})
\ \le\
K_2\,(1+\log \budget)^{\frac{1}{1-2\alpha}}+K_3.
\end{equation}
Since $\alpha\in(0,1/2)$ is fixed, there exists an integer $m_2\ge 2$ such that for all $\budget\ge m_2$,
\[
K_2\,(1+\log \budget)^{\frac{1}{1-2\alpha}}+K_3
\ \le\
\frac{\budget}{4C_{\max}}.
\]
Let $m$ be any integer satisfying $m\ge \max\{4C_{\max}(t_0+1),\,m_2\}$.
Then for all $\budget\ge m$, we have $n-t_0\ge \budget/(4C_{\max})\ge m_1(\budget^{-2})$ and $\budget\ge 4C_{\max}(t_0+1)$.
Applying the explicit $(\budget,\delta)$ bound with $\delta=\budget^{-2}$ yields
\[
\EE[\mathrm{SR}(\tau)]
\ \le\
\frac{2\sqrt{C_{\max}}\,\bar{\mathsf A}(\budget^{-2})}{\sqrt{\budget}}
+\Delta_C\,\budget^{-2}
+\Delta_C\,\exp\!\left(-\frac{\budget\,C_{\max}}{4R_c^2}\right),
\]
where the exponential term is omitted when $R_c=0$.
Finally, $\bar{\mathsf A}(\budget^{-2})=O(\sqrt{\log \budget})$ with a hidden constant independent of $\budget$,
and hence for all $\budget\ge m$,
\[
\EE[\mathrm{SR}(\tau)]
=
O\!\left(\sqrt{\frac{\log \budget}{\budget}}\right),
\]
which completes the proof.
\qed

\begin{proof}[Proof of Lemma~\ref{lem:max-ineq-cond-gauss}]
Fix $\lambda>0$ and define, for $s\ge 0$,
\[
Z_s \ =\ \exp\!\left(\lambda M_s-\frac{\lambda^2R_c^2}{2}\,s\right),
\qquad\text{with } M_0=0 \text{ and } Z_0=1.
\]
By Assumption~\ref{ass:observation}, for each $s\ge 1$,
\[
\EE\!\left[\exp(\lambda\eta_{c,s})\mid \cF_{s-1}\right]\ \le\ \exp\!\left(\frac{\lambda^2R_c^2}{2}\right),
\]
and hence
\begin{equation*}
\EE[Z_s\mid \cF_{s-1}]
=
Z_{s-1}\cdot
\EE\!\left[\exp\!\left(\lambda\eta_{c,s}-\frac{\lambda^2R_c^2}{2}\right)\Bigm|\cF_{s-1}\right]\le Z_{s-1}.
\end{equation*}
Therefore, $\{Z_s\}_{s\ge 0}$ is a nonnegative supermartingale with $\EE[Z_0]=1$.
By Ville's inequality, for any $a>0$,
\[
\Pr\!\left[\max_{1\le s\le n} Z_s\ge a\right]\ \le\ \frac{\EE[Z_0]}{a}\ =\ \frac{1}{a}.
\]
If $\max_{1\le s\le n} M_s>x$, then there exists $s\in[n]$ such that $M_s>x$, and thus
\[
Z_s
=\exp\!\left(\lambda M_s-\frac{\lambda^2R_c^2}{2}\,s\right)
\ \ge\
\exp\!\left(\lambda x-\frac{\lambda^2R_c^2}{2}\,n\right).
\]
Consequently,
\begin{align*}
\Pr\!\left[\max_{1\le s\le n} M_s>x\right]
&\le
\Pr\!\left[\max_{1\le s\le n} Z_s\ge \exp\!\left(\lambda x-\frac{\lambda^2R_c^2}{2}\,n\right)\right]\\
&\le
\exp\!\left(-\lambda x+\frac{\lambda^2R_c^2}{2}\,n\right).
\end{align*}
Optimizing over $\lambda>0$ by choosing $\lambda=x/(nR_c^2)$ gives
\[
\Pr\!\left[\max_{1\le s\le n} M_s>x\right]\ \le\ \exp\!\left(-\frac{x^2}{2nR_c^2}\right),
\]
which completes the proof.
\end{proof}

\section{Tightness Results}\label{appendix:impossibility}
To show that the dependence $\Omega(1/(\Delta^\star)^2)$ cannot be removed from any convergent simple-regret guarantee under Problem~\ref{problem-main}, where $\Delta^\star=-g(\btheta^\star)$, we construct a set of counterexamples.
Let $Q=1$, $\cQ=\{q\}$, $\Theta=\{\btheta_0,\btheta_1\}$, $s(\btheta_0)=1$, $\epsilon=1/2$, $c(\btheta_0)=1$ and $c(\btheta_1)=0$, so $s_0=1/2$.
For each $\Delta\in(0,1/2]$, define two instances
\[
\cI^+_\Delta:s(\btheta_1)=1/2+\Delta,\quad \cI^-_\Delta:s(\btheta_1)=1/2-\Delta.
\]
Assume that the observations on $g$ satisfy $y_{g,t}=g(\btheta_t)+\eta_t$ with $\eta_t\sim\cN(0,1)$, and the observations on $c$ satisfy $y_{c,t}=c(\btheta_t)$.
One can verify that these instances satisfy Assumptions~\ref{ass:observation}--\ref{ass:rkhs} with fixed constants.

Fix $\delta\in(0,1/4)$ and an integer $\tau\in\NN$.
Let $\pi$ be any adaptive algorithm that makes $\tau$ observations and outputs $\btheta_{\rm out}\in\Theta$.
Assume that for every $\Delta\in(0,1/2]$, when the true instance is $\cI^-_\Delta$, the algorithm is $\delta$-correct.
The impossibility theorem follows.

\begin{theorem}
\label{thm:inv-gap2-necessary}
Under the above construction, on $\cI^+_{\Delta}$ the simple regret satisfies
\begin{equation}\label{eq:sr-lb-gap-short}
\EE\!\left[\mathrm{SR}(\tau)\right]\ \ge\ \left(1-\delta-\Delta^\star\sqrt{\tau}\right)_+,
\end{equation}
and in particular, any guarantee of the form $\EE[\mathrm{SR}(\tau)]\le \varepsilon$ with $\varepsilon<1-\delta$ requires
\[
\tau\ \ge\ \frac{(1-\delta-\varepsilon)^2}{(\Delta^\star)^2}
\ =\ \Omega\!\left(\frac{1}{(\Delta^\star)^2}\right).
\]
\end{theorem}

\begin{proof}
Fix $\Delta\in(0,1/2]$ and define the event $A=\{\btheta_{\rm out}=\btheta_1\}$.
On $\cI^-_\Delta$, $\btheta_1$ is infeasible since $g(\btheta_1)=\Delta>0$.
Therefore, $\delta$-correctness implies $\Pr_{\cI^-_\Delta}(A)\le \delta$.

Let $P_{-\Delta}$ and $P_{+\Delta}$ denote the joint laws of the full transcript under $\cI^-_\Delta$ and $\cI^+_\Delta$.
Let $N_1=\sum_{t=1}^{\tau}\bm 1\{\btheta_t=\btheta_1\}$ be the number of evaluations of $\btheta_1$, so $N_1\le \tau$ almost surely.
The two instances differ only in the distribution of $y_{g,t}$ when $\btheta_t=\btheta_1$.
Under $\cI^-_\Delta$, $y_{g,t}\sim\cN(\Delta,1)$, while under $\cI^+_\Delta$, $y_{g,t}\sim\cN(-\Delta,1)$.
Using the chain rule for KL divergence under adaptive sampling~\cite{10.5555/2946645.2946646}, we have
\begin{align*}
\KL(P_{-\Delta}\,\|\,P_{+\Delta})&
=
\EE_{-\Delta}[N_1]\cdot \KL(\cN(\Delta,1)\,\|\,\cN(-\Delta,1))\\
&=
\EE_{-\Delta}[N_1]\cdot 2\Delta^2
\ \le\ 2\tau\Delta^2.
\end{align*}
Pinsker's inequality then gives
\[
\|P_{-\Delta}-P_{+\Delta}\|_{\rm TV}
\ \le\ \sqrt{\KL(P_{-\Delta}\,\|\,P_{+\Delta})/2}
\ \le\ \Delta\sqrt{\tau}.
\]
Therefore
\[
\Pr_{\cI^+_\Delta}(A)
\le
\Pr_{\cI^-_\Delta}(A)+\|P_{-\Delta}-P_{+\Delta}\|_{\rm TV}
\le
\delta+\Delta\sqrt{\tau},
\]
and hence $\Pr_{\cI^+_\Delta}(A^c)\ge (1-\delta-\Delta\sqrt{\tau})_+$.

On $\cI^+_\Delta$, $\btheta_1$ is feasible and optimal since $c(\btheta_1)=0<c(\btheta_0)=1$, so $\btheta^\star=\btheta_1$ and $\Delta^\star=-g(\btheta^\star)=\Delta$.
Thus $\mathrm{SR}(\tau)=1$ exactly on $A^c$, and $\EE[\mathrm{SR}(\tau)]=\Pr(A^c)$.
Substituting $\Delta^\star=\Delta$ yields \eqref{eq:sr-lb-gap-short}.
Rearranging $(1-\delta-\Delta^\star\sqrt{\tau})_+\le \varepsilon$ with $\varepsilon<1-\delta$ gives the stated necessary condition on $\tau$.
\end{proof}

\end{document}